\documentclass[11pt,a4paper]{article}

\usepackage[utf8]{inputenc}
\usepackage[T1]{fontenc}
\usepackage{lmodern}
\usepackage[margin=1in]{geometry}
\usepackage{microtype} 
\usepackage{amsmath,amssymb,amsthm,amsfonts}
\usepackage{mathtools}
\usepackage{bm}         
\usepackage{bbm}         
\usepackage{braket}       
\newcommand{\openone}{\mathbbm{1}}
\newcommand{\ii}{\mathrm{i}}

\theoremstyle{plain}
\newtheorem{proposition}{Proposition}[section]
\newtheorem{theorem}{Theorem}[section]
\newtheorem{corollary}[theorem]{Corollary}
\usepackage{graphicx}
\usepackage{subcaption}  
\usepackage{booktabs}     
\usepackage{tikz}
\usetikzlibrary{quantikz2} 
\usepackage{authblk}     
\usepackage{xcolor}
\definecolor{quantumblue}{RGB}{83,37,127} 
\usepackage[
  colorlinks=true,
  linkcolor=quantumblue,
  citecolor=quantumblue,
  urlcolor=quantumblue,
  breaklinks=true
]{hyperref}
\usepackage[numbers,sort&compress]{natbib}
\title{Quantum Reference Frames in Quantum Circuits: Perspective-Dependent Entangling Cost and Coherence–Entanglement Trade-Offs}
\author[1]{Salman Sajad Wani}
\author[1]{Saif Al-Kuwari}
\affil[1]{Qatar Center for Quantum Computing, College of Science and Engineering, Hamad Bin Khalifa University, Doha, Qatar}

\begin{document}
\maketitle

\begin{abstract}
The perspective-neutral formulation of quantum reference frames (QRFs) treats observers as quantum systems and describes physics relationally from within the composite system. While frame-change maps and frame-invariant resource sums are theoretically understood, their impact on circuit-based quantum information processing has largely remained unexplored. We formulate QRF transformations as circuit compilation rules and, for systems with finite Abelian symmetry described by the regular representation, derive a gate-level dictionary that maps local operations in one frame to their images in another. This yields a group-theoretic classification of gates where symmetry-commuting operators remain local, up to frame-dependent phases, while generic gates are promoted to controlled entangling operations in which the original frame acts as a control register. The resulting frame-dependence entangling-gate count defines a relational circuit complexity where the cost of a computation depends on the internal reference frame of the observer. We instantiate the framework in a three-qubit model and show that the QRF unitary acts as a lossless converter between a purity-based local coherence and concurrence, preserving their invariant sum and giving a concrete realization of the relativity of entanglement. We implement the corresponding circuits on an IBM Quantum superconducting platform, using full state tomography to reconstruct the redistribution of resources between internal frames. The hardware data reproduce the predicted conversion of local coherence into entanglement, and the observed deviations from exact conservation quantify the effect of realistic device noise on relational quantum protocols.
\end{abstract}

\maketitle

\section{Introduction}

Standard quantum theory is usually formulated relative to an external, classical reference frame—an idealized laboratory that is assumed to be infinite, noninteracting, and perfectly stable. Although this framework underpins most quantum experiments and technologies, it fails in regimes where the reference system is small, quantum, or dynamically coupled to what it describes in relativistic quantum information, distributed quantum networks, and high-precision metrology with finite-size apparatuses \cite{Aharonov1984QuantumFrames,Bartlett2007ReferenceFramesRMP,Loveridge2017Relativity,Angelo2012QRFInertial,Rovelli1991ReferenceSystems,Rovelli1996RelationalQM,Poulin2006RelationalToyModel}. The perspective-neutral program of quantum reference frames (QRFs) addresses this by treating reference frames as physical quantum systems and formulating physics relationally: a single global state admits different internal descriptions, each associated with a different quantum observer \cite{Giacomini2019CovarianceQRF,Vanrietvelde2020ChangePerspective,DeLaHamette2020QRFGeneralGroups,Palmer2014ChangingQRF}. Relativistic and dynamical generalizations extend this framework to N-body systems, spin degrees of freedom, and quantum subsystems in gravitational or accelerated settings, revealing frame-dependent observables and symmetries \cite{Vanrietvelde2020NBODY,Krumm2021QRFTransformations,Ballesteros2021GroupStructureQRF,Giacomini2019RelativisticSpin,Mikusch2021SpinQRF,Ahmad2022QuantumRelativity,ahmad2022relativity,Cepollaro2023GravTimeDilation}.

The kinematic structure of QRFs—how states and observables transform between internal observers—is characterized as a unitary change of relational basis generated by an underlying symmetry \cite{Vanrietvelde2020ChangePerspective,DeLaHamette2020QRFGeneralGroups,Hoehn2020RelationalClocks}. Much of the literature, however, treats frame changes as abstract isometries or channels acting on constrained Hilbert spaces, with an emphasis on covariance properties, relational observables, and links to quantum gravity and relativistic quantum information \cite{Bartlett2007ReferenceFramesRMP,Giacomini2019CovarianceQRF,Ahmad2022QuantumRelativity,ahmad2022relativity,Hoehn2020RelationalClocks,Krumm2021QRFTransformations,Ballesteros2021GroupStructureQRF}. By contrast, the operational consequences for quantum information processing—gate locality, entangling cost, and resource accounting—remain underexplored. Given a standard universal gate set, an important question is how each elementary gate is represented in a different internal frame, when a local gate remains local or becomes controlled, and how entangling structure and circuit depth are redistributed between frames. Addressing these questions is necessary if QRF ideas are to inform algorithm design, resource theories of asymmetry and coherence \cite{baumgratz_quantifying_2014,streltsov_coherence_review_2017,Chitambar2019ResourceTheories,GourSpekkens2008QRFResource,Horodecki2009Entanglement}, and implementations on near-term devices.

In this work, we take an operational view and treat QRF transformations as circuit rewrite rules. We show that the behavior of a gate under a change of frame is determined by its group-theoretic covariance: gates that commute with the group action remain local, while generic gates are compiled into controlled or genuinely entangling operations in which the old reference frame acts as a control register. This circuit-level perspective connects QRF transformations—originally formulated as Hilbert-space isometries and relational observables~\cite{Bartlett2007ReferenceFramesRMP,Giacomini2019CovarianceQRF,Vanrietvelde2020ChangePerspective,DeLaHamette2020QRFGeneralGroups,Palmer2014ChangingQRF,Vanrietvelde2020NBODY,Krumm2021QRFTransformations,Ballesteros2021GroupStructureQRF}—to standard notions of entangling-gate cost and relates the frame dependence of coherence and entanglement to gate-level structure.

Theoretically, we establish a gate-level QRF calculus for finite Abelian groups. Starting from the perspective-neutral change-of-frame map for finite groups~\cite{Vanrietvelde2020ChangePerspective,DeLaHamette2020QRFGeneralGroups,Krumm2021QRFTransformations,Ballesteros2021GroupStructureQRF}, we then derive a transformation law for single-qudit gates acting on systems that carry the regular representation of $G$. This yields a classification of one-qudit operations: symmetry-commuting gates remain local, character-sector gates acquire only frame-dependent phases, and generic gates become genuinely entangling between the frame and the subsystem. In particular, we show that a local superposition in one frame (e.g.\ a Hadamard on a qubit) becomes a controlled-entangling operation in another, providing an operational account of the frame dependence of locality in the language of quantum logic gates. This classification also underpins a notion of relational circuit complexity where the minimal entangling-gate count of a fixed global unitary becomes frame dependent, with an overhead bounded by the number of noncovariant local gates in a chosen decomposition.
We then apply this framework to the trade-off between coherence and entanglement. Building on the resource-theoretic analysis of Cepollaro {et al.}~\cite{Cepollaro2025InvariantSum} and on known complementarity relations for two-qubit pure states~\cite{Englert1996Duality,Bera2015Duality,Wootters1998Concurrence,Horodecki2009Entanglement,baumgratz_quantifying_2014}, we consider a three-qubit $\mathbb{Z}_2$ model and track bipartite entanglement and local single-qubit coherence across different internal frames. For the symmetric family we consider, an appropriate coherence–entanglement sum is conserved, while the frame change converts local coherence in one perspective into bipartite entanglement in another. In the experimental sections, we access this conservation law through a purity-based local coherence measure extracted from tomography on the frame qubit.
Finally, we develop and implement a three-qubit protocol that realizes these ideas as a digital quantum simulation of a finite-group QRF transformation. Unlike continuous-variable QRF and relational models, which often encounter factorization subtleties or require unbounded energy resources~\cite{Ahmad2022QuantumRelativity,DeLaHamette2020QRFGeneralGroups,Hammerer2010Trajectories,Zeuthen2022TrajectoriesComposite,Cepollaro2023GravTimeDilation}, our finite-group construction leads to a finite-dimensional, tensor-factorizable Hilbert space that can be directly mapped to qubits. We specify a three-qubit circuit, built from single-qubit Clifford gates, CNOTs, and SWAPs, that prepares a relational state, implements the $\mathbb{Z}_2$ frame-change unitary as a fixed gate sequence, and performs tomography on selected subsystems to reconstruct $C^2$ and a local coherence measure $D^2$ before and after the frame change. We validate the protocol using density-matrix simulations (Qiskit Aer~\cite{Qiskit}) and execute it on an IBM superconducting device~\cite{iBMQ}. We observe the predicted redistribution of local coherence and bipartite entanglement, confirming the conservation of the invariant resource sum in ideal simulations, and quantifying deviations in the experimental data that are consistent with decoherence and gate infidelity on current NISQ hardware~\cite{Preskill2018NISQ}.

Throughout, we restrict attention to finite Abelian groups and to few-qubit circuits, with the three-qubit $\mathbb{Z}_2$ model providing a minimal setting in which all constructions are explicit and experimentally testable. Within this scope, our results move quantum reference frames from abstract Hilbert-space maps to experimentally implementable circuit identities that directly probe QRF-induced resource trade-offs and relational circuit complexity on near-term processors.

The rest of this paper is organized as follows. Sections~\ref{sec:finite-group-QRF} and \ref{sec:finiteG-gate-calculus} introduce the finite-group QRF framework and derive a gate-level transformation dictionary; for Abelian groups, we obtain a three-class distinction between frame-robust, phase-only, and genuinely entangling gates, and we illustrate it in a three-qubit $\mathbb{Z}_2$ model. Section~\ref{sec:circuits} extends this dictionary to full circuits and introduces a relational entangling-gate complexity with a corresponding overhead bound. Section~\ref{sec:nisq-protocol} implements the $\mathbb{Z}_2$ frame change as a shallow three-qubit circuit and benchmarks the predicted coherence--entanglement redistribution in noiseless simulation and on IBM Quantum hardware using state tomography. Technical derivations are collected in Appendix~\ref{app:finiteG}--Appendix~\ref{app:Z2-operator-identities}.

\section{Finite-Group Quantum Reference Frames}
\label{sec:finite-group-QRF}

In the perspective-neutral formulation of quantum reference frames (QRFs), physics is described globally on a constraint surface of the total Hilbert space, where choosing a frame is implemented as a unitary transformation to a relational perspective~\cite{Giacomini2019CovarianceQRF,Vanrietvelde2020ChangePerspective,DeLaHamette2020QRFGeneralGroups}. Different internal observers are related by frame-change isometries that, in finite-dimensional models, extend to unitary operators on the kinematical tensor-product space. We focus on to systems transforming under a finite group $G$. Each subsystem $X$ has a basis $\{\ket{g}_X\}_{g\in G}$ that transforms under the right-regular representation, $U_R(h)\ket{g}_X = \ket{g h^{-1}}_X$. The total Hilbert space consists of an initial reference frame $0$, a target frame $i$, and a set of registers $R$, so that
\begin{equation}
  \mathcal{H}_{\mathrm{tot}} = \mathcal{H}_0 \otimes \mathcal{H}_i \otimes \bigotimes_{k\in R} \mathcal{H}_k .
\end{equation}
The change of perspective from frame $0$ to frame $i$ is implemented by the unitary~\cite{Vanrietvelde2020ChangePerspective}
\begin{equation}
  U_{0\to i}
  =
  \mathrm{SWAP}_{0,i}
  \sum_{g\in G}
  \ket{g}\!\bra{g}_i
  \otimes \openone_0
  \otimes
  \bigotimes_{k\in R} U_R(g)_k,
  \label{eq:finite-group-QRF-unitary}
\end{equation}
where $\mathrm{SWAP}_{0,i}$ exchanges the frame registers and the controlled operation applies the symmetry transformation $g$ to the remaining systems relative to the new orientation of frame $i$. This map is the quantum analogue of the coordinate shift $x \to x - x_i$ (see Appendix~\ref{app:finiteG-U} for the derivation from the constraint formalism).

We work in the Heisenberg picture to describe how such frame changes act on quantum circuits. Physical operations—measurements, observables, and gates—transform according to
\begin{equation}
  \hat{O}^{(j)}
  =
  U_{i\to j}\,\hat{O}^{(i)}\,U_{i\to j}^\dagger,
  \label{eq:obs-transform}
\end{equation}
which preserves expectation values by construction. Our primary interest lies in the application of this map to quantum logical gates. Consider a gate acting locally on subsystem $S$ in frame $i$, represented globally as $\hat{O}^{(i)} = \openone_{\overline{S}} \otimes \hat{U}_S$. As detailed in the following section, the conjugation in Eq.~\eqref{eq:obs-transform} then provides the dictionary for circuit compilation and, in general, maps local operations to controlled-entangling gates.

\section{Gate-level quantum reference frames for finite Abelian groups}
\label{sec:finiteG-gate-calculus}

Building on the operator map in Eq.~\eqref{eq:obs-transform}, we now pass to the circuit level. We ask: given an operation local to subsystem $S$ in frame $0$, does its image in frame $i$ remain local, or does the frame change turn it into an entangling controlled operation? We first treat a general finite group $G$. Let a local gate acting on subsystem $S\in R$ in frame $0$ be represented globally as $\hat{O}_S^{(0)} = \openone_{\overline{S}} \otimes \hat{U}_S$. Its representation in frame $i$ is given by the Heisenberg conjugate $\hat{O}_S^{(i)} = U_{0\to i}\,\hat{O}_S^{(0)}\,U_{0\to i}^\dagger$. The following theorem shows that this transformation always maps local gates to controlled-unitary operations, with the old frame serving as the control register.

\begin{theorem}[Finite-group gate-transform]
\label{thm:finiteG-gate-transform-main}
Let $G$ be a finite group acting via the right-regular representation $U_R$. Under the change of reference frame $0\to i$ defined by Eq.~\eqref{eq:finite-group-QRF-unitary}, the image of a local unitary $\hat{U}_S$ is
\begin{equation}
  \hat{O}_S^{(i)}
  =
  \sum_{g\in G}
  \ket{g}\!\bra{g}_{\,0}
  \otimes
  \bigl(U_R(g)\,\hat{U}_S\,U_R(g)^\dagger\bigr)_S
  \otimes
  \openone_{\overline{\{0,S\}}}.
  \label{eq:finiteG-gate-transform-main}
\end{equation}
\end{theorem}

Physically, the classical value $g$ stored in the old frame register $0$ controls which conjugated version of the gate is applied to $S$. The proof (Appendix~\ref{app:finiteG-theorem}) uses the decomposition $U_{0\to i} = \mathrm{SWAP}_{0,i}\,W_i$, where conjugation by $W_i$ dresses the gate with group actions $U_R(g)$ and the subsequent SWAP transfers the control role to the initial frame. This general transformation rule simplifies significantly for finite Abelian groups. In this case, the conjugation action $\alpha_g(\cdot) \coloneqq U_R(g)(\cdot)U_R(g)^\dagger$ defines a commuting family of automorphisms, and the structure of the transformed gate is determined by the orbit $\mathcal{O}_{\hat{U}} = \{\alpha_g(\hat{U}_S)\}_{g\in G}$ of the target unitary. This yields a trichotomy for the behaviour of local gates:

\begin{corollary}[Classification of Abelian gate transformations]
\label{cor:abelian-classification}
Let $G$ be Abelian. Then the transformed gate $\hat{O}_S^{(i)}$ falls into exactly one of the following three cases:
\begin{enumerate}
  \item[(i)]  If $[\hat{U}_S,U_R(g)]=0$ for all $g\in G$, then
  $\hat{O}_S^{(i)}=\openone_0\otimes\hat{U}_S$. The gate remains strictly local on $S$ and decouples from the frame.
  \item[(ii)]  If $\hat{U}_S$ is an eigenoperator of the action $\alpha_g$, i.e., $\alpha_g(\hat{U}_S)=\chi(g)\hat{U}_S$ for a character $\chi$, then
  $\hat{O}_S^{(i)}=V_0(\chi)\otimes\hat{U}_S$, where $V_0(\chi)=\sum_g \chi(g)\ket{g}\!\bra{g}_0$. The frame change induces a local control phase but no entanglement.
  \item[(iii)]  If the orbit $\mathcal{O}_{\hat{U}}$ contains operators that are not related by a global phase, then $\hat{O}_S^{(i)}$ is an entangling controlled unitary across the $0{:}S$ partition.
\end{enumerate}
\end{corollary}

This classification provides a group-theoretic criterion for relational compilation: only gates in the commutant of the regular representation are robust against frame changes, while generic operations are compiled into entangling resources. To illustrate the three cases explicitly, consider the qubit example with $G = \mathbb{Z}_2$. The regular representation is generated by the Pauli operator $U_R(1) = X$, and the gate transformation in Eq.~\eqref{eq:finiteG-gate-transform-main} takes the form (see Appendix~\ref{app:finiteG-Z2} for the derivation as a special case of Theorem~\ref{thm:finiteG-gate-transform-main})
\begin{equation}
  \hat{O}_S^{(i)}
  =
  \ket{0}\!\bra{0}_0 \otimes \hat{U}_S
  +
  \ket{1}\!\bra{1}_0 \otimes \bigl(X \hat{U}_S X\bigr).
  \label{eq:Z2-gate-transform}
\end{equation}
Combining Corollary~\ref{cor:abelian-classification} with Eq.~\eqref{eq:Z2-gate-transform}, the three cases for $G=\mathbb{Z}_2$ are:
\begin{enumerate}
  \item[(i)] Frame-robust case: If $[\hat{U}_S,X]=0$, for example for rotations generated by $X$, $R_x(\theta) = \exp(-\mathrm{i}\theta X/2)$, then $\hat{O}_S^{(i)} = \openone_0 \otimes \hat{U}_S$. A local $R_x$ gate in frame $0$ therefore remains local in frame $i$ and decouples from the reference frame.
  \item[(ii)] Phase-control case: If $\{\hat{U}_S,X\}=0$, for example for the Pauli-$Z$ gate, we have $XZX = -Z$, and Eq.~\eqref{eq:Z2-gate-transform} gives $\hat{O}_S^{(i)} = Z_0 \otimes Z_S$. A local phase operation is promoted to a tensor-product operation controlled by the frame register.
  \item[(iii)] Entangling case: In the generic situation the operators $\hat{U}_S$ and $X \hat{U}_S X$ are not related by a global phase. For the Hadamard gate $H$ one finds $X H X = (X - Z)/\sqrt{2}$, which is linearly independent of $H$, so a local superposition gate in frame $0$ is compiled into an entangling controlled unitary $C(\hat{U}_S, X \hat{U}_S X)_{0S}$ in frame $i$.
\end{enumerate}

\subsection{Model definition and relational perspectives}
\label{subsec:model-def}

Consider three qubits, labeled $A, B$, and $C$. Following the setup in Sec.~\ref{sec:finiteG-gate-calculus}, the group $G=\mathbb{Z}_2$ acts on each qubit via the Pauli-$X$ operator. The total Hilbert space is $\mathcal{H} \cong (\mathbb{C}^2)^{\otimes 3}$. We restrict to the relational subspace satisfying the global charge constraint $\Pi_{\mathrm{phys}} \ket{\psi} = \ket{\psi}$, where the projector is defined by the parity operator
\begin{equation}
  \Pi_{\mathrm{phys}}
  \;=\;
  \frac{1}{2}\bigl( \openone + Z_A Z_B Z_C \bigr).
\end{equation}
The physical subspace $\mathcal{H}_{\mathrm{phys}}$ is therefore spanned by the basis states with even total parity ($a \oplus b \oplus c = 0$):
\begin{equation}
  \mathcal{H}_{\mathrm{phys}}
  \;=\;
  \mathrm{span}\bigl\{
    \ket{000}, \ket{011}, \ket{101}, \ket{110}
  \bigr\}.
  \label{eq:physical-subspace}
\end{equation}
A general pure state in this subspace is a superposition of these basis vectors,
\begin{equation}
  \ket{\Psi}_{ABC}
  \;=\;
  \alpha \ket{000} + \beta \ket{011} + \gamma \ket{101} + \delta \ket{110},
\end{equation}
subject to the normalization $|\alpha|^2 + |\beta|^2 + |\gamma|^2 + |\delta|^2 = 1$. We treat $C$ as the initial laboratory frame. A change of perspective from $C$ to an internal frame $F\in\{A,B\}$ is implemented by the finite-group QRF unitary $U_{C\to F}$ derived in Eq.~\eqref{eq:finite-group-QRF-unitary}. For the $\mathbb{Z}_2$ action, the frame register plays the role of control and the remaining systems the role of targets, and we define the state relative to frame $F$ as
\begin{equation}
  \ket{\Psi}^{(F)} \;\coloneqq\; U_{C\to F}\ket{\Psi}.
\end{equation}
The explicit matrix representations of $U_{C\to F}$ on $\mathcal{H}_{\mathrm{phys}}$ and the verification of subspace invariance are given in Appendix~\ref{app:Z2-matrices}, while Appendix~\ref{app:Z2-operator-identities} contains the detailed derivation of the $\mathbb{Z}_2$ frame-change unitaries and the corresponding operator identities.

To obtain the description from a given frame, we decompose the global state in that frame's computational basis and trace out the frame register. In the initial laboratory frame $C$, the state admits the decomposition
\begin{equation}
  \ket{\Psi}_{ABC}
  \;=\;
  \sum_{k=0}^1 \ket{k}_C \otimes \ket{\chi_k}_{AB},
  \label{eq:chi-decomposition}
\end{equation}
where the conditional states $\ket{\chi_k}_{AB}$ encode the dependence of the system $AB$ on the configuration of the frame. The explicit amplitudes in terms of $\alpha, \beta, \gamma, \delta$ are listed in Appendix~\ref{app:Z2-matrices}. Applying the QRF unitaries $U_{C\to A}$ and $U_{C\to B}$ yields the corresponding descriptions from the perspectives of frames $A$ and $B$:
\begin{align}
  \ket{\Psi}^{(A)}
  &\,=\,
  \sum_{k=0}^1 \ket{k}_A \otimes \ket{\phi_k^{(A)}}_{BC},
  \label{eq:psi-A-decomp} \\
  \ket{\Psi}^{(B)}
  &\,=\,
  \sum_{k=0}^1 \ket{k}_B \otimes \ket{\phi_k^{(B)}}_{AC}.
  \label{eq:psi-B-decomp}
\end{align}
Here, the vectors $\ket{\phi_k^{(F)}}$ represent the conditional state of the remaining systems given that the reference frame $F$ is in state $\ket{k}$.

The relational description seen by an observer in frame $F$ is given by the reduced density matrix of the remaining degrees of freedom. Tracing out the reference system in each of the three perspectives yields
\begin{align}
  \rho^{(C)}_{AB} \;&=\; \sum_{k} \ket{\chi_k}\!\bra{\chi_k}_{AB}, \label{eq:rho-AB-C} \\
  \rho^{(A)}_{BC} \;&=\; \sum_{k} \ket{\phi_k^{(A)}}\!\bra{\phi_k^{(A)}}_{BC}, \label{eq:rho-BC-A} \\
  \rho^{(B)}_{AC} \;&=\; \sum_{k} \ket{\phi_k^{(B)}}\!\bra{\phi_k^{(B)}}_{AC}. \label{eq:rho-AC-B}
\end{align}
These reduced states $\rho^{(F)}$ contain all information accessible to the internal observer $F$ and will serve as the starting point for quantifying the frame dependence of coherence and entanglement.

\subsection{Resource conservation and trade-offs}
\label{subsec:coh-ent-diagnostics}

To quantify how quantum resources are redistributed between frames, we track the trade-off between the bipartite entanglement of an effective system pair and the local coherence of a distinguished subsystem. We quantify entanglement using the squared Wootters concurrence $C^2(\rho)$~\cite{Wootters1998Concurrence}. For a pure two-qubit state $\ket{\psi}_{XY}$, this reduces to the linear entropy of the marginals,
\begin{equation}
  C^2(\ket{\psi}_{XY}) \;=\; 4 \det(\rho_X).
\end{equation}
Locally, we decompose the state of a single qubit $S$ using the Bloch vector $\bm{r}=(r_x, r_y, r_z)$, such that $\rho_S = \frac{1}{2}(\openone + \bm{r}\cdot\bm{\sigma})$. We identify two complementary local resources:
\begin{align}
  D^2(\rho_S) &\;\coloneqq\; r_x^2 + r_y^2 = 4|\rho_{01}|^2, \label{eq:D2-def}\\
  P^2(\rho_S) &\;\coloneqq\; r_z^2 = (\rho_{00} - \rho_{11})^2. \label{eq:P2-def}
\end{align}
Here, $D^2$ measures the local coherence (squared $\ell_1$-norm) in the computational basis, while $P^2$ quantifies the population predictability.

For pure bipartite states, the entanglement implies mixedness in the subsystems ($C^2 = 1 - \|\bm{r}_X\|^2$), leading to the complementarity relation~\cite{jakob2010Complementarity}
\begin{equation}
  C^2(\rho_{XY}) \;+\; D^2(\rho_X) \;+\; P^2(\rho_X) \;=\; 1.
  \label{eq:full-complementarity}
\end{equation}
A self-contained derivation of $C^2 + D^2 + P^2 = 1$ and its specialization to $C^2 + D^2 = 1$ for pure two-qubit states is given in Appendix~\ref{app:coh_conc}. In the three-qubit $\mathbb{Z}_2$ model, we focus on a symmetric family of states where the local populations of the singled-out qubit are frame-invariant. Specifically, for states where the local subsystem is maximally mixed in the $Z$-basis ($P^2=0$), Eq.~\eqref{eq:full-complementarity} simplifies to the frame-invariant conservation law
\begin{equation}
  C_F^2 \;+\; D_F^2 \;=\; 1, \qquad F \in \{A,B,C\}.
  \label{eq:C2+D2=1}
\end{equation}
This relation captures the QRF resource conversion: a fixed quantum resource appears either as local superposition or as nonlocal entanglement, depending on the internal observer's perspective.

We now apply these diagnostics to the symmetric three-qubit model. For each internal frame $F \in \{A,B,C\}$, we identify a system pair $S_{\text{pair}}$ for entanglement monitoring and a local subsystem $S_{\text{loc}}$ for coherence monitoring. We track a one-parameter trajectory of physical states 
$\ket{\Psi(\lambda)}$, whose explicit form is given in Appendix~\ref{app:coh_conc}, and use the coherence–concurrence complementarity derived there. Evaluating the reduced states for each observer yields an explicit realization of the complementarity relation. While the partition of resources is frame-dependent, the total resource $C_F^2(\lambda) + D_F^2(\lambda)$ remains invariant across all three perspectives:
\begin{equation}
  C_F^2(\lambda) \;+\; D_F^2(\lambda) \;=\; 1,
  \qquad \forall \lambda, \quad \forall F \in \{A,B,C\}.
  \label{eq:C2+D2=1-lambda}
\end{equation}
This result confirms that the QRF transformation $U_{C\to F}$ behaves as a lossless resource converter. The specific choices of subsystems and the resulting resource flows are summarized in Table~\ref{tab:C2-D2-frames}. This analytical prediction yields a clear experimental signature where a continuous redistribution of $C^2$ and $D^2$ that always sums to unity. We implement this protocol on superconducting hardware in Sec.~\ref{sec:nisq-protocol}.

\begin{table}[t]
  \centering
  \caption{Frame-dependent resource distribution. For each frame $F$, we define the bipartite cut for concurrence ($C_F^2$) and the subsystem for local coherence ($D_F^2$). For fixed $F$ and $\lambda$, the conservation law Eq.~\eqref{eq:C2+D2=1-lambda} ensures that any decrease in local coherence is exactly compensated by an increase in entanglement.}
  \label{tab:C2-D2-frames}
  \begin{tabular}{c c c l}
    \toprule
    Frame ($F$) & Pair ($S_{\text{pair}}$) & Local ($S_{\text{loc}}$) & Resource Behavior \\
    \midrule
    $C$ & $A$--$B$ & $B$ & Entanglement increases with $\lambda$ \\
    $A$ & $B$--$C$ & $B$ & Coherence dominates; $C_A^2 \to 0$ \\
    $B$ & $A$--$C$ & $A$ & Intermediate exchange \\
    \bottomrule
  \end{tabular}
\end{table}

The three-qubit $\mathbb{Z}_2$ model provides a concrete operational realization of the relativity of entanglement. The conservation law in Eq.~\eqref{eq:C2+D2=1-lambda} shows that a change of internal frame acts as a resource interconversion: the QRF unitary redistributes a fixed quantum budget between local superposition and nonlocal correlations. Crucially, the identity of the resource-bearing subsystem is frame-dependent. A feature that manifests as bipartite entanglement to observer $C$ may be compiled into local coherence for observer $A$. This gives a precise operational meaning to the statement suggesting that entanglement is a relational attribute rather than an absolute property of the state~\cite{Bartlett2007ReferenceFramesRMP,Giacomini2019CovarianceQRF}. Furthermore, the restriction to finite groups circumvents the technical pathologies associated with continuous-variable QRFs, such as factorization anomalies or non-normalizable reference states~\cite{Ahmad2022QuantumRelativity,DeLaHamette2020QRFGeneralGroups}. In our model, the frame change is a well-defined unitary on a finite-dimensional, tensor-factorizable Hilbert space, making the QRF transformation directly accessible as a gate sequence on digital quantum processors. Consequently, the choice of internal frame becomes a programmable design parameter—a feature we exploit in the experimental protocol of Sec.~\ref{sec:nisq-protocol} to probe the resource trade-off on NISQ hardware.


\section{Quantum circuits in different reference frames}
\label{sec:circuits}

 A unitary process decomposed in frame $i$ as a sequence of elementary gates, $V_{\mathrm{proc}}^{(i)} = U_L \cdots U_1$, transforms into frame $j$ via global conjugation. The linearity of the QRF map implies that the transformation distributes over the gate sequence:
\begin{equation}
  V_{\mathrm{proc}}^{(j)}
  \;=\;
  U_{i\to j} \left( \prod_{k=1}^L U_k^{(i)} \right) U_{i\to j}^\dagger
  \;=\;
  \prod_{k=1}^L \left( U_{i\to j} \, U_k^{(i)} \, U_{i\to j}^\dagger \right).
  \label{eq:circuit-compilation}
\end{equation}
Equation~\eqref{eq:circuit-compilation} provides a concrete procedure for relational circuit compilation. While the global dynamics represent a passive transformation (invariant physics), the circuit topology is frame-dependent. Applying the dictionary derived in Sec.~\ref{sec:finiteG-gate-calculus} re-compiles a circuit of local gates in frame $i$ into a sequence of potentially entangling operations in frame $j$. The structural complexity of the circuit—defined by its entangling-gate count—is therefore relative to the observer.

\subsection{Relational compilation in the \texorpdfstring{$\mathbb{Z}_2$}{Z2} model}
\label{subsec:relational-compilation-Z2}

We first use Eq.~\eqref{eq:circuit-compilation} on the three-qubit $\mathbb{Z}_2$ model and a standard Bell-preparation circuit defined in frame $C$. The process $V_{\text{Bell}}^{(C)}$ acts on the input $\ket{000}$ via a Hadamard on $A$ followed by a $\mathrm{CNOT}_{A\to B}$, preparing the state $\ket{\Phi^+}_{AB}\otimes\ket{0}_C$. This textbook circuit is shown in Fig.~\ref{fig:canonical-circuit-three-frames}(a), where the frame qubit $C$ remains idle.

To describe this same physical process from the perspective of internal frame $B$, we compile the unitary $V_{\text{Bell}}^{(C)}$ using the conjugation rule
\begin{equation}
  U \;\mapsto\; U_{C\to B}\, U \, U_{C\to B}^\dagger .
\end{equation}
Using the gate dictionary from Sec.~\ref{sec:finiteG-gate-calculus}, we find that the CNOT gate is frame-robust because it commutes with the relevant parity operator. In contrast, the local Hadamard gate $H_A$ is generic: in frame $B$ it appears as an entangling controlled operation between the frame register $B$ and the system qubit $A$ (see Appendix~\ref{app:Z2-operator-identities}). Figure~\ref{fig:canonical-circuit-three-frames}(b) shows the compiled circuit, in which the frame qubit actively participates in the dynamics. Similarly, moving to frame $A$ produces a circuit in which the entangling gate acts on the $BC$ pair; the corresponding decomposition is shown in Fig.~\ref{fig:canonical-circuit-three-frames}(c). Taken together, these examples show that changing the QRF rewrites the circuit: the global unitary remains the same, but its decomposition into local and entangling gates depends on the observer.

\begin{figure*}[t]
  \centering
  \subcaptionbox{Frame $C$: standard Bell preparation on $AB$. The frame $C$ is idle.}[0.3\textwidth]{%
    \begin{quantikz}
      \lstick{$C$ (frame)} & \qw & \qw & \qw \\
      \lstick{$A$} & \gate{H} & \ctrl{1} & \qw \\
      \lstick{$B$} & \qw & \targ{} & \qw
    \end{quantikz}
  }\hspace{1em}
  \subcaptionbox{Frame $B$: the local $H_A$ is compiled into a controlled gate with the frame register $B$ as control.}[0.3\textwidth]{%
    \begin{quantikz}
      \lstick{$B$ (frame)} & \ctrl{1} & \qw & \qw \\
      \lstick{$A$} & \gate{U_H} & \ctrl{1} & \qw \\
      \lstick{$C$} & \qw & \targ{} & \qw
    \end{quantikz}
  }\hspace{1em}
  \subcaptionbox{Frame $A$: the entangling resources are shifted to the $BC$ subsystem, controlled by frame $A$.}[0.3\textwidth]{%
    \begin{quantikz}
      \lstick{$A$ (frame)} & \ctrl{1} & \qw \\
      \lstick{$B$} & \gate{U_{BC}} & \qw \\
      \lstick{$C$} & \qw & \qw
    \end{quantikz}
  }
  \caption{Relational circuit compilation for the $\mathbb{Z}_2$ model. The same physical process (Bell preparation) is decomposed into different gate sequences depending on the observer's reference frame. The local resources in frame $C$ (a) are recompiled into entangling resources involving the frame register in frames $B$ (b) and $A$ (c).}
  \label{fig:canonical-circuit-three-frames}
\end{figure*}
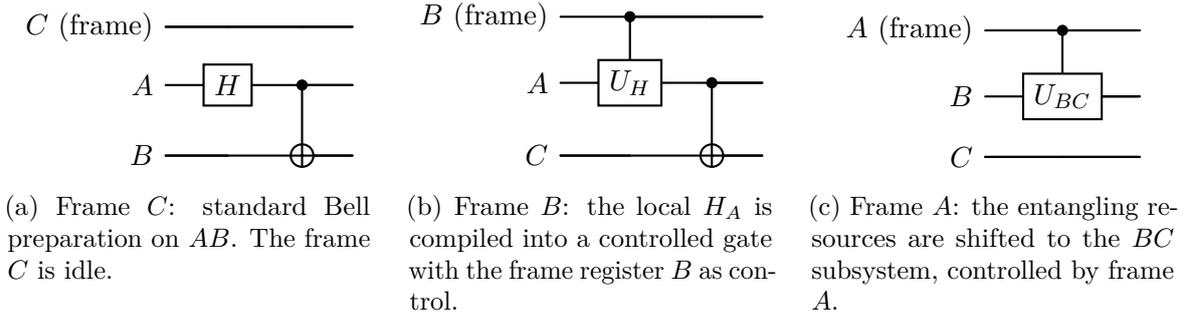

The Bell example illustrates a general pattern. For a multi-qubit subsystem \(S\), the global group action is implemented by the parity operator
\begin{equation}
  X_S \coloneqq \bigotimes_{q\in S} X_q.
\end{equation}
Under a frame change \(0\to i\), a local gate \(U_S\) transforms according to the multi-qubit generalization of Eq.~\eqref{eq:Z2-gate-transform} (see Appendix~\ref{app:finiteG-Z2} for the derivation as a direct application of the finite-group gate-transform theorem to composite subsystems):
\begin{equation}
  U_S^{(i)}
  =
  \ket{0}\!\bra{0}_0 \otimes U_S
  +
  \ket{1}\!\bra{1}_0 \otimes (X_S U_S X_S).
  \label{eq:multi-qubit-transform}
\end{equation}
This induces a simple compilation dictionary for standard circuit elements. The \(\mathrm{CNOT}_{AB}\) gate commutes with the global parity operator \(X_A \otimes X_B\) (commuting an \(X\) through the target of a CNOT induces an \(X\) on the control), so CNOT gates remain local in all frames. Likewise, any phase rotation generated by Pauli-\(X\) or products thereof (such as \(R_{xx}(\theta)\)) commutes with \(X_S\) and is frame-robust. By contrast, as established in Sec.~\ref{sec:finiteG-gate-calculus}, the Hadamard gate \(H\) does not commute with \(X\) and transforms into a controlled operation \(C(H,Z)_{0S}\), which entangles the frame register with the system. This distinction—that CNOTs and \(X\)-generated phases are frame-invariant while Hadamards are not—drives the structural rearrangement of the Bell circuit observed in Fig.~\ref{fig:canonical-circuit-three-frames} and underlies the notion of relational circuit complexity in Sec.~\ref{subsec:relational-complexity}.

\subsection{Relational circuit complexity}
\label{subsec:relational-complexity}

The frame-dependence of the gate decomposition motivates the notion of {relational circuit complexity}. Since global symmetries are fixed, we quantify complexity via the count of two-qubit entangling primitives, $N_{\mathrm{ent}}$.
Consider a fixed decomposition of the processing unitary $U_{\mathrm{proc}}$ in frame $0$. We partition the set of local single-qubit gates into covariant operations $\mathcal{G}_{\mathrm{cov}}$, which satisfy $X U X^\dagger = e^{\ii\phi} U$, and generic operations $\mathcal{G}_{\mathrm{gen}}$, where the conjugation yields a linearly independent unitary. Under a frame change $0\to i$, elements of $\mathcal{G}_{\mathrm{cov}}$ remain local (up to control phases), while elements of $\mathcal{G}_{\mathrm{gen}}$ turn into controlled-unitary gates that entangle the frame register with the system. Consequently, the entangling cost in the new frame is bounded by
\begin{equation}
  N_{\mathrm{ent}}^{(i)}
  \;\le\;
  N_{\mathrm{ent}}^{(0)} \;+\; |\mathcal{G}_{\mathrm{gen}}|.
  \label{eq:entangling-overhead-bound}
\end{equation}
This inequality quantifies the additional entangling cost induced by changing frames: the overhead is strictly determined by the number of symmetry-breaking local gates in the original decomposition. The Bell-preparation circuit (Fig.~\ref{fig:canonical-circuit-three-frames}) saturates this bound. In frame $C$, the circuit uses one entangling gate (CNOT) and one generic local gate (Hadamard), so $N_{\mathrm{ent}}^{(C)}=1$ and $|\mathcal{G}_{\mathrm{gen}}|=1$.
In frames $A$ and $B$, the Hadamard is compiled into a controlled operation, raising the cost to $N_{\mathrm{ent}}^{(A,B)} = 2$.
For larger architectures, this suggests a relationally aware compilation strategy: reducing the use of non-covariant gates in the logical circuit reduces the worst-case entangling overhead across different internal perspectives.

\section{Experimental validation}
\label{sec:nisq-protocol}


In this section we designate frame $A$ as the laboratory frame (in contrast to Sec.~\ref{sec:circuits}, where $C$ played the role of an external frame for the Bell circuit). The initial physical state is prepared from the vacuum state $\ket{000}_{ABC}$ via the circuit $V_{\mathrm{prep}} = X_C H_B$, yielding
\begin{equation}
  \ket{\Psi_{\mathrm{lab}}}
  \;=\;
  \ket{0}_A \ket{+}_B \ket{1}_C.
  \label{eq:Psi-lab-def}
\end{equation}
Using the diagnostics established in Sec.~\ref{subsec:coh-ent-diagnostics}, this state exhibits maximal local coherence on subsystem $C$ ($D_C^2=1$) and vanishing entanglement across the $A$--$C$ partition ($C_{AC}^2=0$).

To switch to the perspective of internal frame $B$, we apply the frame-change unitary $U_{A\to B}$ obtained from the finite-group QRF map in Eq.~\eqref{eq:finite-group-QRF-unitary}. For the specific subspace on which the initial frame $A$ is fixed in state $\ket{0}$, the general map $U_{A\to B}$ compiles into a compact gate sequence:
\begin{equation}
  U_{A\to B}
  \;\cong\;
  \mathrm{SWAP}_{A,B} \cdot \mathrm{CNOT}_{B\to C},
  \label{eq:U-A-to-B-experiment}
\end{equation}
(see Appendix~\ref{app:Z2-operator-identities} for the derivation from the general finite-group map).
Acting on $\ket{\Psi_{\mathrm{lab}}}$, this transformation produces the state
\begin{equation}
  \ket{\Psi_{\mathrm{qrf}}}
  \;=\;
  \frac{1}{\sqrt{2}}\bigl(\ket{001} + \ket{100}\bigr).
\end{equation}

From the perspective of frame $B$, obtained by tracing out the reference register $B$ from $\ket{\Psi_{\mathrm{qrf}}}$, the resource distribution is inverted: subsystem $C$ becomes maximally mixed ($\rho_C^{(B)} = \openone/2 \implies D_C^2=0$), while the $A$--$C$ pair becomes a maximally entangled Bell state ($C_{AC}^2=1$). We numerically simulated this protocol using Qiskit Aer~\cite{Qiskit}. The results confirm the exact conservation of the total resource sum between the two perspectives:
\begin{equation}
  \Sigma_{\text{res}}^{(B)} - \Sigma_{\text{res}}^{(A)}
  \;=\;
  (C_{AC}^2 + D_C^2)_B - (C_{AC}^2 + D_C^2)_A
  \;=\;
  (1 + 0) - (0 + 1)
  \;=\; 0.
\end{equation}
This ideal response establishes a baseline for hardware performance. Since the protocol uses only three qubits and a shallow depth ($\sim 4$ CNOTs including SWAP decomposition), it is well-suited to probing resource redistribution on NISQ devices.
\subsection{Experimental setup}
\label{sec:exp-methods}

The protocol was executed on an IBM Quantum superconducting processor. The hardware requirements are modest: only three qubits with linear connectivity and a standard Clifford+CNOT gate set are needed. The state preparation and frame-change unitary $U_{A\to B}$ [Eq.~\eqref{eq:U-A-to-B-experiment}] give a total circuit depth of four CNOTs (including the SWAP decomposition), which is shallow enough to preserve coherence on present-day NISQ devices.

To reconstruct the resource measures, we performed full quantum state tomography (QST) on the three-qubit system in both frames. This required $3^3=27$ Pauli measurement settings per frame, for a total of 54 configurations. With a shot budget of $N_{\text{shots}} \sim 10^3$ per basis, the statistical uncertainty is negligible compared to hardware gate errors. While partial tomography of the marginals $\rho_C$ and $\rho_{AC}$ would slightly reduce the measurement overhead, full reconstruction allows us to monitor global purity and verify that the dynamics remain confined to the physical subspace.

The experimental signature of the QRF transformation is the predicted inversion of resources: a successful implementation should show a statistically significant decrease in local coherence $D_C^2$ together with the appearance of concurrence $C_{AC}^2$. Deviations from the conservation law $C^2+D^2=1$ serve as a benchmark for device noise (specifically decoherence and state mixing) rather than a breakdown of the relational formalism.
\subsection{Experimental results}
\label{sec:exp-results}

We benchmark the protocol on three platforms: theoretical analysis , noiseless simulation (Qiskit Aer), and the 27-qubit IBM superconducting device \texttt{ibm\_fez}. Table~\ref{tab:theta-pi-over-2} summarizes the reconstructed resource measures for the equal-superposition state ($\theta=\pi/2$). In the noiseless simulation, the reconstructed values track the theoretical predictions with high fidelity. In frame $A$, we find $D_C^2 \approx 0.987(2)$ and $C_{AC}^2 \approx 0$, confirming the separability of the lab-frame state. In frame $B$, the roles of the resources are reversed, with $D_C^2 \approx 0$ and $C_{AC}^2 \approx 0.973(6)$. The total resource sum $C^2+D^2 \approx 0.98$ deviates from unity by an amount consistent with finite-sampling bias ($N_{\text{shots}}=10^3$) inherent to maximum-likelihood estimation.

On the \texttt{ibm\_fez} device, the qualitative signature of resource redistribution is still visible despite hardware noise. Frame $A$ exhibits high local coherence ($D_C^2 \approx 0.96$) and negligible entanglement. After the frame change to $B$, the local coherence vanishes ($D_C^2 \approx 0.00$) and significant entanglement emerges ($C_{AC}^2 \approx 0.74$).
The reduction of the total invariant to $\approx 0.74$ in frame $B$ reflects the decoherence accumulated during the four-CNOT frame-change circuit. However, the relative weights of local coherence and entanglement clearly invert, indicating that the QRF unitary interconverts coherence and entanglement as expected.

\begin{table}[t]
  \centering
  \caption{Experimental validation of resource redistribution. We compare analytic predictions with results from noiseless simulation (\texttt{Aer}) and superconducting hardware (\texttt{ibm\_fez}). Uncertainties represent the standard deviation over repeated tomography runs. The hardware results demonstrate a clear inversion of resources between frames $A$ and $B$, despite the non-unit total sum due to decoherence.}
  \label{tab:theta-pi-over-2}
  \begin{tabular}{llccc}
    \toprule
    Frame & Platform & Local Coherence ($D_C^2$) & Concurrence ($C_{AC}^2$) & Total ($D^2+C^2$) \\
    \midrule
    $A$ & Theory & 1.000 & 0.000 & 1.000 \\
        & \texttt{Aer} & $0.987 \pm 0.002$ & $0.000 \pm 0.000$ & $0.988$ \\
        & \texttt{ibm\_fez} & $0.963 \pm 0.005$ & $0.000 \pm 0.000$ & $0.963$ \\
    \midrule
    $B$ & Theory & 0.000 & 1.000 & 1.000 \\
        & \texttt{Aer} & $0.000 \pm 0.000$ & $0.973 \pm 0.006$ & $0.973$ \\
        & \texttt{ibm\_fez} & $0.001 \pm 0.000$ & $0.743 \pm 0.032$ & $0.744$ \\
    \bottomrule
  \end{tabular}
\end{table}

We illustrate the complexity bound in Eq.~\eqref{eq:entangling-overhead-bound} using the Bell-preparation circuit defined in Sec.~\ref{sec:circuits}. In the laboratory frame $C$, the decomposition $V_{\mathrm{proc}}^{(C)} = \mathrm{CNOT}_{AB} H_A$ contains exactly one entangling primitive ($N_{\mathrm{ent}}^{(C)} = 1$) and one symmetry-breaking local gate ($H_A \in \mathcal{G}_{\mathrm{gen}}$). As shown in Sec.~\ref{subsec:relational-complexity}, the transformation to frame $B$ promotes the Hadamard to a controlled operation. The explicit circuit in Fig.~\ref{fig:canonical-circuit-three-frames}(b) then requires two entangling gates ($N_{\mathrm{ent}}^{(B)} = 2$), saturating the relational overhead bound
\begin{equation}
  N_{\mathrm{ent}}^{(B)}
  \;=\;
  N_{\mathrm{ent}}^{(C)} \;+\; |\mathcal{G}_{\mathrm{gen}}|
  \;=\; 2.
\end{equation}

To check that this complexity gap is intrinsic to the algorithm and not an artifact of the logical gate set, we transpiled both frame descriptions to the native basis of the \texttt{ibm\_fez} device (echoed cross-resonance). We found that the physical two-qubit gate count remained invariant under compilation: the frame-$C$ circuit compiled to a single hardware entangler, while the frame-$B$ circuit required two. Taken together, these observations show that the entangling cost of relationality is a structural property of the circuit: the difference in entangling gate count between frames $C$ and $B$ survives compilation to the native hardware gate set.


\section{Conclusion}
\label{sec:conclusion}

In this paper, we make the perspective-neutral framework of quantum reference frames operational by treating frame changes as circuit compilation rules. Considering systems transforming under finite Abelian groups, we derive a gate-level dictionary that maps local operations in one frame to their images in another. This construction connects the abstract kinematics of QRFs~\cite{Vanrietvelde2020ChangePerspective,DeLaHamette2020QRFGeneralGroups,Giacomini2019CovarianceQRF,Krumm2021QRFTransformations,Ballesteros2021GroupStructureQRF} to the concrete gate sequences used in quantum information processing. The main message is that the locality of a quantum gate is fixed by its group-theoretic covariance. For Abelian symmetries, we obtain a classification of gates into three classes: frame-robust operations in the commutant of the representation, character-sector gates that acquire only frame-dependent phases, and generic gates that are compiled into controlled-entangling operations. This mechanism realises the ``relativity of entanglement'' as a trade-off between local coherence and nonlocal correlations, consistent with invariant-sum relations between coherence and entanglement under QRF changes~\cite{Cepollaro2025InvariantSum,Horodecki2009Entanglement,Bera2015Duality,Englert1996Duality}. In the $\mathbb{Z}_2$ model, we show explicitly that this trade-off obeys the conservation law $C^2 + D^2 = 1$ across all internal perspectives. Our hardware implementation on an IBM Quantum processor demonstrates that QRF transformations can be implemented and probed on present-day NISQ devices~\cite{Preskill2018NISQ,Kandala2017HardwareEfficient,Vandersypen2005NMRQC}.

These results open several directions for future work. First, the notion of relational circuit complexity---where the entangling cost of a computation depends on the observer---suggests a resource theory combining ideas from asymmetry and coherence~\cite{baumgratz_quantifying_2014,streltsov_coherence_review_2017,Chitambar2019ResourceTheories} with the resource theory of quantum reference frames~\cite{GourSpekkens2008QRFResource}. Such a framework could optimize quantum algorithms by selecting the reference frame that minimizes entangling gate counts. Second, extending the gate calculus to non-Abelian groups or continuous symmetries (e.g., $\mathrm{SU}(2)$) is key to applying these tools to relativistic quantum information and clock-synchronization problems~\cite{Ahmad2022QuantumRelativity,ahmad2022relativity,Giacomini2019CovarianceQRF,Mikusch2021SpinQRF}. Finally, the interplay between frame-dependent resources and metrological precision remains largely unexplored. Identifying whether relational entanglement offers advantages in multi-parameter estimation would link QRFs to quantum sensing and relativistic metrology applications~\cite{Hammerer2010Trajectories,Zeuthen2022TrajectoriesComposite,Cepollaro2023GravTimeDilation}.

\appendix
\bibliographystyle{unsrt}      
\bibliography{main}

\begin{thebibliography}{99}

\bibitem{Aharonov1984QuantumFrames}
Y.~Aharonov and T.~Kaufherr,
``Quantum frames of reference,''
Phys.\ Rev.\ D \textbf{30}, 368--385 (1984),
doi:10.1103/PhysRevD.30.368.

\bibitem{Rovelli1991ReferenceSystems}
C.~Rovelli,
``Quantum reference systems,''
Class.\ Quantum Grav.\ \textbf{8}, 317--332 (1991),
doi:10.1088/0264-9381/8/2/011.

\bibitem{Rovelli1996RelationalQM}
C.~Rovelli,
``Relational quantum mechanics,''
Int.\ J.\ Theor.\ Phys.\ \textbf{35}, 1637--1678 (1996),
doi:10.1007/BF02302261.

\bibitem{Bartlett2007ReferenceFramesRMP}
S.~D.~Bartlett, T.~Rudolph, and R.~W.~Spekkens,
``Reference frames, superselection rules, and quantum information,''
Rev.\ Mod.\ Phys.\ \textbf{79}, 555--609 (2007),
doi:10.1103/RevModPhys.79.555.

\bibitem{Giacomini2019CovarianceQRF}
F.~Giacomini, E.~Castro-Ruiz, and \v{C}.~Brukner,
``Quantum mechanics and the covariance of physical laws in quantum reference frames,''
Nat.\ Commun.\ \textbf{10}, 494 (2019),
doi:10.1038/s41467-018-08155-0.

\bibitem{Vanrietvelde2020ChangePerspective}
A.~Vanrietvelde, P.~A.~H{\"o}hn, E.~Castro-Ruiz, and F.~Giacomini,
``A change of perspective: Switching quantum reference frames via a perspective-neutral framework,''
Quantum \textbf{4}, 225 (2020),
doi:10.22331/q-2020-01-27-225.

\bibitem{DeLaHamette2020QRFGeneralGroups}
A.~C.~de~la~Hamette and T.~D.~Galley,
``Quantum reference frames for general symmetry groups,''
Quantum \textbf{4}, 367 (2020),
doi:10.22331/q-2020-11-30-367.

\bibitem{Ahmad2022QuantumRelativity}
A.~S.~Ahmad, T.~D.~Galley, P.~A.~H{\"o}hn, M.~P.~E.~Lock, and A.~R.~H.~Smith,
``Quantum relativity of subsystems,''
Phys.\ Rev.\ Lett.\ \textbf{128}, 170401 (2022),
doi:10.1103/PhysRevLett.128.170401.

\bibitem{ahmad2022relativity}
H.~Ahmad, F.~Giacomini, W.~Layton, A.~Smith, and \v{C}.~Brukner,
``Relativity of subsystems and the Schr{\"o}dinger equation in quantum reference frames,''
Quantum \textbf{6}, 836 (2022),
doi:10.22331/q-2022-11-10-836.

\bibitem{Cepollaro2025InvariantSum}
C.~Cepollaro, A.~Akil, P.~Cie\'sli\'nski, A.~C.~de~la~Hamette, and \v{C}.~Brukner,
``Sum of entanglement and subsystem coherence is invariant under quantum reference frame transformations,''
Phys.\ Rev.\ Lett.\ \textbf{135}, 010201 (2025),
doi:10.1103/PhysRevLett.135.010201.

\bibitem{Loveridge2017Relativity}
L.~Loveridge, P.~Busch, and T.~Miyadera,
``Relativity of quantum states and observables,''
EPL \textbf{117}, 40004 (2017),
doi:10.1209/0295-5075/117/40004.

\bibitem{Angelo2012QRFInertial}
R.~M.~Angelo and A.~D.~Ribeiro,
``Quantum reference frames and inertial motion,''
Phys.\ Rev.\ A \textbf{85}, 052109 (2012),
doi:10.1103/PhysRevA.85.052109.

\bibitem{Vanrietvelde2020NBODY}
A.~Vanrietvelde, P.~A.~H{\"o}hn, and F.~Giacomini,
``Switching quantum reference frames in the N-body problem and the absence of global relational perspectives,''
Quantum \textbf{7}, 1088 (2023),
doi:10.22331/q-2023-08-22-1088.

\bibitem{baumgratz_quantifying_2014}
T.~Baumgratz, M.~Cramer, and M.~B.~Plenio,
``Quantifying coherence,''
Phys.\ Rev.\ Lett.\ \textbf{113}, 140401 (2014),
doi:10.1103/PhysRevLett.113.140401.

\bibitem{streltsov_coherence_review_2017}
A.~Streltsov, G.~Adesso, and M.~B.~Plenio,
``Colloquium: Quantum coherence as a resource,''
Rev.\ Mod.\ Phys.\ \textbf{89}, 041003 (2017),
doi:10.1103/RevModPhys.89.041003.

\bibitem{fan_universal_complementarity_2018}
X.-G.~Fan, W.-Y.~Sun, Z.-Y.~Ding, F.~Ming, H.~Yang, D.~Wang, and L.~Ye,
``Universal complementarity between coherence and intrinsic concurrence for two-qubit states,''
New J.\ Phys.\ \textbf{21}, 093053 (2019),
doi:10.1088/1367-2630/ab41b1.

\bibitem{zhou_mutual_restriction_2020}
A.-L.~Zhou, D.~Wang, X.-G.~Fan, F.~Ming, and L.~Ye,
``Mutual restriction between concurrence and intrinsic concurrence for arbitrary two-qubit states,''
Chin.\ Phys.\ Lett.\ \textbf{37}, 110301 (2020),
doi:10.1088/0256-307X/37/11/110302.

\bibitem{Horodecki2009Entanglement}
R.~Horodecki, P.~Horodecki, M.~Horodecki, and K.~Horodecki,
``Quantum entanglement,''
Rev.\ Mod.\ Phys.\ \textbf{81}, 865--942 (2009),
doi:10.1103/RevModPhys.81.865.

\bibitem{Wootters1998Concurrence}
W.~K.~Wootters,
``Entanglement of formation of an arbitrary state of two qubits,''
Phys.\ Rev.\ Lett.\ \textbf{80}, 2245--2248 (1998),
doi:10.1103/PhysRevLett.80.2245.

\bibitem{Preskill2018NISQ}
J.~Preskill,
``Quantum computing in the NISQ era and beyond,''
Quantum \textbf{2}, 79 (2018),
doi:10.22331/q-2018-08-06-79.

\bibitem{Kandala2017HardwareEfficient}
A.~Kandala, A.~Mezzacapo, K.~Temme, M.~Takita, M.~Brink, J.~M.~Chow, and J.~M.~Gambetta,
``Hardware-efficient variational quantum eigensolver,''
Nature \textbf{549}, 242--246 (2017),
doi:10.1038/nature23879.

\bibitem{Vandersypen2005NMRQC}
L.~M.~K.~Vandersypen and I.~L.~Chuang,
``NMR techniques for quantum control and computation,''
Rev.\ Mod.\ Phys.\ \textbf{76}, 1037--1069 (2005),
doi:10.1103/RevModPhys.76.1037.

\bibitem{Qiskit}
M.~S.~Anis, H.~Abraham, \textit{et al.},
``Qiskit: An open-source framework for quantum computing,''
Zenodo (2021),
doi:10.5281/zenodo.2562110.

\bibitem{iBMQ}
IBM Quantum,
\url{https://quantum.ibm.com},
accessed 2025-12-05 (2025).

\bibitem{hoehn2020RelationalClocks}
P.~A.~H{\"o}hn and A.~Vanrietvelde,
``How to switch between relational quantum clocks,''
Phys.\ Rev.\ D \textbf{102}, 046012 (2020),
doi:10.1103/PhysRevD.102.046012.

\bibitem{Chitambar2019ResourceTheories}
E.~Chitambar and G.~Gour,
``Quantum resource theories,''
Rev.\ Mod.\ Phys.\ \textbf{91}, 025001 (2019),
doi:10.1103/RevModPhys.91.025001.

\bibitem{Englert1996Duality}
B.-G.~Englert,
``Fringe visibility and which-way information: An inequality,''
Phys.\ Rev.\ Lett.\ \textbf{77}, 2154--2157 (1996),
doi:10.1103/PhysRevLett.77.2154.

\bibitem{Bera2015Duality}
M.~N.~Bera, T.~Qureshi, M.~A.~Siddiqui, and A.~K.~Pati,
``Duality of quantum coherence and path distinguishability,''
Phys.\ Rev.\ A \textbf{92}, 012118 (2015),
doi:10.1103/PhysRevA.92.012118.

\bibitem{jakob2010Complementarity}
M.~Jakob and J.~A.~Bergou,
``Quantitative complementarity relations in bipartite systems,''
Opt.\ Commun.\ \textbf{283}, 827--830 (2010),
doi:10.1016/j.optcom.2009.10.044.

\bibitem{Palmer2014ChangingQRF}
T.~Palmer, F.~Girelli, and S.~D.~Bartlett,
``Changing quantum reference frames,''
Phys.\ Rev.\ A \textbf{89}, 052121 (2014),
doi:10.1103/PhysRevA.89.052121.

\bibitem{Krumm2021QRFTransformations}
M.~Krumm, P.~A.~H{\"o}hn, and M.~P.~M{\"u}ller,
``Quantum reference frame transformations as symmetries and the paradox of the third particle,''
Quantum \textbf{5}, 530 (2021),
doi:10.22331/q-2021-08-27-530.

\bibitem{Giacomini2019RelativisticSpin}
F.~Giacomini, E.~Castro-Ruiz, and \v{C}.~Brukner,
``Relativistic quantum reference frames: The operational meaning of spin,''
Phys.\ Rev.\ Lett.\ \textbf{123}, 090404 (2019),
doi:10.1103/PhysRevLett.123.090404.

\bibitem{Cepollaro2023GravTimeDilation}
C.~Cepollaro, F.~Giacomini, and M.~G.~A.~Paris,
``Gravitational time dilation as a resource in quantum sensing,''
Quantum \textbf{7}, 946 (2023),
doi:10.22331/q-2023-03-13-946.

\bibitem{Ballesteros2021GroupStructureQRF}
A.~Ballesteros, F.~Giacomini, and G.~Gubitosi,
``The group structure of dynamical transformations between quantum reference frames,''
Quantum \textbf{5}, 470 (2021),
doi:10.22331/q-2021-06-08-470.

\bibitem{GourSpekkens2008QRFResource}
G.~Gour and R.~W.~Spekkens,
``The resource theory of quantum reference frames: Manipulations and monotones,''
New J.\ Phys.\ \textbf{10}, 033023 (2008),
doi:10.1088/1367-2630/10/3/033023.

\bibitem{Poulin2006RelationalToyModel}
D.~Poulin,
``Toy model for a relational formulation of quantum theory,''
Int.\ J.\ Theor.\ Phys.\ \textbf{45}, 1189--1215 (2006),
doi:10.1007/s10773-006-9195-2.

\bibitem{Hammerer2010Trajectories}
J.~F.~Sherson and K.~Hammerer,
``Trajectories without quantum uncertainties in the presence of entanglement,''
Phys.\ Rev.\ A \textbf{82}, 021803 (2010),
doi:10.1103/PhysRevA.82.021803.

\bibitem{Zeuthen2022TrajectoriesComposite}
E.~Zeuthen, E.~S.~Polzik, and F.~Ya.~Khalili,
``Trajectories without quantum uncertainties in composite systems with disparate energy spectra,''
Phys.\ Rev.\ A \textbf{106}, L010601 (2022),
doi:10.1103/PhysRevA.106.L010601.

\bibitem{Mikusch2021SpinQRF}
M.~Mikusch, M.~Krumm, F.~Giacomini, and \v{C}.~Brukner,
``Transformation of spin in quantum reference frames,''
Phys.\ Rev.\ Res.\ \textbf{3}, 043138 (2021),
doi:10.1103/PhysRevResearch.3.043138.

\bibitem{sun_intrinsic_relations_2017}
W.-Y.~Sun, D.~Wang, B.-L.~Fang, Z.-Y.~Ding, H.~Yang, F.~Ming, and L.~Ye,
``The intrinsic relations of quantum resources in multiparticle systems,''
Phys.\ Rev.\ A \textbf{96}, 022329 (2017),
doi:10.1103/PhysRevA.96.022329.

\end{thebibliography}

\section{Finite-group QRF map and gate-transform theorem}
\label{app:finiteG}

In this appendix, we collect the technical details underlying the finite-group quantum reference frame (QRF) calculus used in Sec.~\ref{sec:finiteG-gate-calculus}. We first recall the standard perspective-neutral change-of-frame unitary for finite groups and check that it implements the classical relational map. We then derive the finite-group gate-transform theorem, which states that a local gate in one frame becomes a controlled gate in another, with the old frame playing the role of a control register. Finally, we consider finite Abelian groups and to the $\mathbb{Z}_2$ model used in our circuit analysis.

Throughout, we follow the notation and conventions of Refs.~\cite{Vanrietvelde2020ChangePerspective,DeLaHamette2020QRFGeneralGroups,Giacomini2019CovarianceQRF}.

\subsection{Setting and notation}
\label{app:finiteG-setup}

Let $G$ be a finite group with identity element $e$ and order $|G|$. Each physical system $k$ carries a copy of the (right) regular representation of $G$, acting on a Hilbert space
\begin{equation}
\mathcal{H}_k \cong \mathbb{C}^{|G|},
\qquad
\mathcal{H} \;=\; \bigotimes_{k=0}^{n-1} \mathcal{H}_k ,
\end{equation}
with computational basis $\{\ket{g}_k : g\in G\}$ satisfying
\begin{equation}
U_R(h)\ket{g}_k = \ket{g h^{-1}}_k, \qquad \forall\,g,h\in G.
\end{equation}
We label one distinguished system as the old reference frame $0$, a second as the new frame $i$, and the remaining systems as ``particles'' or registers
\begin{equation}
R := \{1,\dots,n-1\}\setminus\{i\}.
\end{equation}

In the perspective-neutral construction~\cite{Vanrietvelde2020ChangePerspective,DeLaHamette2020QRFGeneralGroups}, the relational description of the composite system in frame $0$ is obtained by imposing the gauge condition that system $0$ sits at the group identity, so that a relational basis state takes the form
\begin{equation}
\ket{\vec g}^{(0)} \;=\;
\ket{e}_0 \bigotimes_{k\in R\cup\{i\}} \ket{g_k^{(0)}}_k,
\qquad
g_k^{(0)}\in G.
\end{equation}
In a description relative to frame $i$, the same physical configuration is represented by a state in which system $i$ is at the identity and all other labels are shifted by the group element $g_i^{(0)}$ that previously described the position of $i$ in frame $0$. Classically, the frame change
$0\to i$ acts on the tuple of group elements by
\begin{equation}
\Phi_{0\to i}:\ 
\bigl(e,\, g_i^{(0)},\, (g_k^{(0)})_{k\in R}\bigr)
\longmapsto
\bigl( (g_i^{(0)})^{-1},\, e,\, (g_k^{(0)}(g_i^{(0)})^{-1})_{k\in R} \bigr),
\label{eq:classical-relational-map}
\end{equation}
that is, the new frame is fixed at the identity, the old frame label becomes $(g_i^{(0)})^{-1}$, and all other systems are right-multiplied by $(g_i^{(0)})^{-1}$.

Quantum mechanically, this classical rule is implemented by a unitary $U_{0\to i}$ on $\mathcal{H}$ that maps relational states in frame $0$ to relational states in frame $i$ and acts linearly on superpositions~\cite{Vanrietvelde2020ChangePerspective,Giacomini2019CovarianceQRF}.

\subsection{Finite-group QRF unitary}
\label{app:finiteG-U}

We now recall the standard operator form of the finite-group QRF transformation and check that it implements the classical map~\eqref{eq:classical-relational-map} on the non-reference systems.

\begin{proposition}[Finite-group QRF unitary]
\label{prop:finiteG-U}
Let each system carry the right-regular representation of a finite group $G$. The coherent change of reference frame from $0$ to $i$ is implemented by the unitary
\begin{equation}
U_{0\to i}
\;=\;
\mathrm{SWAP}_{0,i}
\Biggl(
\sum_{g\in G}
\ket{g}\!\bra{g}_{\,i}
\;\otimes\;
\openone_{0}
\;\otimes\;
\bigotimes_{k\in R} U_R(g)_k
\Biggr),
\label{eq:U-finiteG-def}
\end{equation}
where $R=\{1,\dots,n-1\}\setminus\{i\}$ and $U_R(g)_k$ denotes the right-regular representation on system $k$.
On the relational basis states $\ket{\vec g}^{(0)}$ it acts as
\begin{equation}
U_{0\to i}
\Bigl(
\ket{e}_0 \ket{g_i^{(0)}}_i \bigotimes_{k\in R} \ket{g_k^{(0)}}_k
\Bigr)
=
\ket{g_i^{(0)}}_0 \ket{e}_i \bigotimes_{k\in R} \ket{g_k^{(0)} (g_i^{(0)})^{-1}}_k,
\end{equation}
i.e., the new frame is fixed at the identity, and all non-reference systems are right-multiplied by $(g_i^{(0)})^{-1}$.
\end{proposition}

\begin{proof}
Define
\begin{equation}
W_i
\;:=\;
\sum_{g\in G}
\ket{g}\!\bra{g}_{\,i}
\;\otimes\;
\openone_{0}
\;\otimes\;
\bigotimes_{k\in R} U_R(g)_k,
\qquad
U_{0\to i}=\mathrm{SWAP}_{0,i}\,W_i.
\end{equation}
Unitarity of each $U_R(g)$ and orthogonality of the projectors $\{\ket{g}\!\bra{g}_i\}$ imply $W_i^\dagger W_i = W_i W_i^\dagger = \openone$, so $W_i$ and hence $U_{0\to i}$ are unitary.

Acting on a relational basis state, we first apply $W_i$:
\begin{align}
W_i
\Bigl(
\ket{e}_0 \ket{g_i^{(0)}}_i \bigotimes_{k\in R} \ket{g_k^{(0)}}_k
\Bigr)
&=
\sum_{g\in G}
\ket{g}\!\bra{g}_{\,i}
\ket{g_i^{(0)}}_i
\ket{e}_0
\bigotimes_{k\in R} U_R(g)_k \ket{g_k^{(0)}}_k
\nonumber\\
&=
\ket{g_i^{(0)}}_i \ket{e}_0
\bigotimes_{k\in R} U_R(g_i^{(0)})_k \ket{g_k^{(0)}}_k
\nonumber\\
&=
\ket{g_i^{(0)}}_i \ket{e}_0
\bigotimes_{k\in R} \ket{g_k^{(0)}(g_i^{(0)})^{-1}}_k,
\end{align}
where in the last step we used the definition of the right-regular representation. Applying the swap $\mathrm{SWAP}_{0,i}$ exchanges systems $0$ and $i$, giving
\begin{equation}
U_{0\to i}
\Bigl(
\ket{e}_0 \ket{g_i^{(0)}}_i \bigotimes_{k\in R} \ket{g_k^{(0)}}_k
\Bigr)
=
\ket{g_i^{(0)}}_0 \ket{e}_i
\bigotimes_{k\in R} \ket{g_k^{(0)}(g_i^{(0)})^{-1}}_k.
\end{equation}
This shows that the quantum map implements the classical relational shift on all non-reference systems, while assigning the label $g_i^{(0)}$ to the old frame. This choice differs from using $(g_i^{(0)})^{-1}$ only by a bijective relabelling $g\mapsto g^{-1}$ on system $0$. Linearity of $U_{0\to i}$ then ensures that superpositions of basis states are transformed coherently.
\end{proof}

The inverse transformation $U_{i\to 0}=U_{0\to i}^\dagger$ has the same structure with $g$ replaced by $g^{-1}$; for finite groups, this is equivalent to re-indexing the sum over the group.

\subsection{Local operator and gate transformation}
\label{app:finiteG-operator}

We now study how a local operator acting on some system $S\in R$ transforms under a change of quantum reference frame. The fundamental observable map is the Heisenberg-picture relation~\cite{Vanrietvelde2020ChangePerspective}
\begin{equation}
Z^{(i)} = U_{0\to i}\, Z^{(0)} \, U_{i\to 0},
\label{eq:heisenberg-map-app}
\end{equation}
which preserves spectra and expectation values across frames.

Let $A_S$ be an operator acting on system $S$ alone, and let
\begin{equation}
\tilde A_S := \openone_{0}\otimes \openone_{i} \otimes
\Biggl(
\bigotimes_{k\in R\setminus\{S\}} \openone_k
\Biggr)
\otimes A_S
\end{equation}
be its embedding into the full Hilbert space. Because $\tilde A_S$ acts trivially on systems $0$ and $i$, it commutes with the swap $\mathrm{SWAP}_{0,i}$:
\begin{equation}
[\tilde A_S,\,\mathrm{SWAP}_{0,i}] = 0.
\end{equation}
Using the factorisation $U_{0\to i}=\mathrm{SWAP}_{0,i} W_i$ from Proposition~\ref{prop:finiteG-U}, we obtain
\begin{align}
A_S^{(i)}
&:= U_{0\to i}\, \tilde A_S\, U_{i\to 0}
= \mathrm{SWAP}_{0,i}\, W_i \tilde A_S W_i^\dagger \,\mathrm{SWAP}_{0,i}.
\label{eq:AS-i-def}
\end{align}
The nontrivial part of the calculation is the conjugation by $W_i$. Using the definition
\begin{equation}
W_i = \sum_{g\in G} \Pi_g^{(i)} \otimes \openone_0 \otimes \bigotimes_{k\in R} U_R(g)_k,
\qquad
\Pi_g^{(i)}:=\ket{g}\!\bra{g}_{\,i},
\end{equation}
we compute
\begin{align}
W_i \tilde A_S W_i^\dagger
&=
\sum_{g,h\in G}
\Pi_g^{(i)} \Pi_h^{(i)}
\otimes \openone_0
\otimes
\Biggl(
\bigotimes_{k\in R\setminus\{S\}} U_R(g)_k \openone_k U_R(h)_k^\dagger
\Biggr)
\otimes U_R(g)_S A_S U_R(h)_S^\dagger
\nonumber\\
&=
\sum_{g\in G}
\Pi_g^{(i)} \otimes \openone_0
\otimes
\Biggl(
\bigotimes_{k\in R\setminus\{S\}} \openone_k
\Biggr)
\otimes U_R(g)_S A_S U_R(g)_S^\dagger,
\label{eq:Wi-AS-Wi}
\end{align}
where in the second line we used $\Pi_g^{(i)}\Pi_h^{(i)}=\delta_{g,h}\Pi_g^{(i)}$ and unitarity of $U_R(g)$ on the spectator systems.

Finally, we conjugate by $\mathrm{SWAP}_{0,i}$. Since $S \neq 0, i$, the swap operation only exchanges the identity on system $0$ with the projector on system $i$:
\begin{equation}
\mathrm{SWAP}_{0,i} \left( \Pi_g^{(i)} \otimes \openone_0 \right) \mathrm{SWAP}_{0,i} = \openone_i \otimes \Pi_g^{(0)}.
\end{equation}
We therefore obtain the following general formula.

\begin{proposition}[Local operator transform]
\label{prop:local-operator}
Let $A_S$ be an operator acting on system $S\in R$, and let $U_{0\to i}$ be the finite-group QRF unitary \eqref{eq:U-finiteG-def}. Then the transformed operator in frame $i$ is
\begin{equation}
A_S^{(i)}
=
\sum_{g\in G}
\Pi_g^{(0)} \otimes \openone_{i} \otimes
\Biggl(
\bigotimes_{k\in R\setminus\{S\}} \openone_k
\Biggr)
\otimes
\Bigl[ U_R(g)_S A_S U_R(g)_S^\dagger \Bigr],
\qquad
\Pi_g^{(0)}:=\ket{g}\!\bra{g}_{\,0}.
\label{eq:AS-i-general}
\end{equation}
In words: in the new frame $i$, the old reference frame $0$ acts as a control register, and the operation on system $S$ is conjugated by the right-regular representation element $U_R(g)_S$ conditioned on the classical value $g$ stored in system $0$.
\end{proposition}

\subsection{Gate-transform theorem for finite groups}
\label{app:finiteG-theorem}

We now state and prove the central theorem used in the main text.

\begin{theorem}[Gate-transform theorem for finite groups]
\label{thm:gate-transform-finiteG}
Let $G$ be a finite group, and let every system carry the right-regular representation $U_R$ of $G$. Consider a change of quantum reference frame $0\!\to\!i$ implemented by $U_{0\to i}$ in \eqref{eq:U-finiteG-def}. Let $U_S$ be a unitary acting on a single system $S\in R$, and let $\tilde U_S$ be its embedding into the full Hilbert space. Then the image of $U_S$ in frame $i$ is the controlled unitary
\begin{equation}
U_S^{(i)}
=
U_{0\to i}\,\tilde U_S\,U_{i\to 0}
=
\sum_{g\in G}
\Pi_g^{(0)} \otimes \openone_{i} \otimes
\Biggl(
\bigotimes_{k\in R\setminus\{S\}} \openone_k
\Biggr)
\otimes
\Bigl[ U_R(g)_S U_S U_R(g)_S^\dagger \Bigr].
\label{eq:US-i-general}
\end{equation}

In particular, the old reference frame $0$ plays the role of a control register, with control basis $\{\ket{g}_0\}_{g\in G}$, and the target operation on $S$ is given by the conjugacy action of the group element $g$ on $U_S$.
\end{theorem}

\begin{proof}
Set $A_S=U_S$ in Proposition~\ref{prop:local-operator}. The embedding of $U_S$ into the full Hilbert space is exactly $\tilde U_S$, and the relation \eqref{eq:heisenberg-map-app} with $Z^{(0)}=\tilde U_S$ gives
\begin{equation}
U_S^{(i)}
=
U_{0\to i}\,\tilde U_S\,U_{i\to 0}
=
\sum_{g\in G}
\Pi_g^{(0)} \otimes
\Biggl(
\bigotimes_{k\in R\setminus\{S\}} \openone_k
\Biggr)
\otimes
\Bigl[ U_R(g)_S U_S U_R(g)_S^\dagger \Bigr],
\end{equation}
which is \eqref{eq:US-i-general}. Each block on the right-hand side is unitary, and the projectors $\{\Pi_g^{(0)}\}$ are orthogonal and complete on system $0$; hence $U_S^{(i)}$ is unitary and of controlled form.
\end{proof}

Equation~\eqref{eq:US-i-general} shows that the gate-transform law is completely determined by the group action on the target system. No additional assumptions beyond the perspective-neutral construction are required: the structure of the QRF map \eqref{eq:U-finiteG-def} fixes the controlled-unitary form uniquely.

\subsection{Classification for finite Abelian groups}
\label{app:finiteG-abelian}

Theorem~\ref{thm:gate-transform-finiteG} holds for arbitrary finite groups $G$. When $G$ is {Abelian}, further structure emerges because the right-regular representation $U_R$ may be simultaneously diagonalised in a Fourier-transformed basis, and conjugation by $U_R(g)$ defines a commuting family of automorphisms on the operator algebra of $\mathcal{H}_S$.

Define the group action on operators by
\begin{equation}
\alpha_g(U) := U_R(g)_S\,U\,U_R(g)_S^\dagger, \qquad g\in G.
\label{eq:alpha-action}
\end{equation}
Because $U_R$ is a representation, we have $\alpha_g\circ\alpha_h = \alpha_{gh}$ and $\alpha_e=\mathrm{id}$. For finite Abelian $G$, the maps $\{\alpha_g\}_{g\in G}$ commute. Embedding $\alpha_g(U_S)$ on the full Hilbert space by identities on all systems other than $S$, the gate-transform theorem \eqref{eq:US-i-general} can then be written as
\begin{equation}
U_S^{(i)} =
\sum_{g\in G}
\Pi_g^{(0)} \otimes
\alpha_g(U_S)_S
\otimes
\openone_{\overline{\{0,S\}}},
\label{eq:US-i-alpha}
\end{equation}
which is a controlled unitary whose `targets' are the images of $U_S$ under the action of the Abelian group.

The following simple corollaries classify the qualitative behaviour.

\begin{corollary}[Symmetry-commuting gates remain local]
\label{cor:commuting-local}
If $U_S$ commutes with the group action, i.e.
\begin{equation}
\alpha_g(U_S)=U_S
\quad\text{for all } g\in G,
\end{equation}
then
\begin{equation}
U_S^{(i)} =
\Biggl(\sum_{g\in G} \Pi_g^{(0)}\Biggr) \otimes U_S \otimes \openone_{\overline{\{0,S\}}}
= \openone_0 \otimes U_S \otimes \openone_{\overline{\{0,S\}}},
\end{equation}
so the gate remains local and completely decouples from the reference-frame register.
\end{corollary}

\begin{corollary}[One-dimensional covariance sector]
\label{cor:covariant-sector}
Suppose there exists a character $\chi: G\to \mathrm{U}(1)$ such that
\begin{equation}
\alpha_g(U_S) = \chi(g)\,U_S
\quad\text{for all } g\in G.
\end{equation}
Then
\begin{equation}
U_S^{(i)}
=
\Biggl(\sum_{g\in G} \chi(g)\,\Pi_g^{(0)}\Biggr) \otimes U_S \otimes \openone_{\overline{\{0,S\}}}
\equiv
V_0(\chi)\otimes U_S \otimes \openone_{\overline{\{0,S\}}},
\end{equation}
where $V_0(\chi)$ is a diagonal unitary on the control register $0$ with eigenvalues $\chi(g)$. In this case the frame change dresses the local gate by a frame-dependent phase on system $0$, but does not create entanglement between $0$ and $S$.
\end{corollary}

\begin{corollary}[Generic gates become entangling]
\label{cor:generic-entangling}
If the orbit of $U_S$ under the Abelian action $\{\alpha_g\}_{g\in G}$ contains at least two nonproportional unitaries (i.e.\ there exist $g,h\in G$ such that $\alpha_g(U_S)$ and $\alpha_h(U_S)$ are not related by a global phase), then the controlled unitary $U_S^{(i)}$ in \eqref{eq:US-i-alpha} is a genuinely entangling gate between system $0$ and system $S$.
\end{corollary}

\begin{proof}
Write $U_S^{(i)}$ in the control basis on system $0$ as in \eqref{eq:US-i-alpha}. Suppose, for the sake of contradiction, that $U_S^{(i)}$ were equivalent to a product of local unitaries on $0$ and $S$, i.e.\ that there exist unitaries $V_0,W_0$ on system $0$ and $V_S,W_S$ on system $S$ such that
\begin{equation}
(V_0\otimes V_S)\,U_S^{(i)}\,(W_0\otimes W_S)
= \openone_0 \otimes U_* \otimes \openone_{\overline{\{0,S\}}}
\end{equation}
for some unitary $U_*$ on $S$. In the control basis this would imply that all blocks on $S$ are equal up to a common conjugation, i.e.\ that
\begin{equation}
\alpha_g(U_S) = e^{\ii\theta_g}\,U_*
\quad\text{for all } g\in G,
\end{equation}
so the entire orbit $\{\alpha_g(U_S): g\in G\}$ consists of unitaries proportional to $U_*$. This contradicts the assumption that the orbit contains two nonproportional unitaries. Hence $U_S^{(i)}$ cannot be written as a product of local unitaries on $0$ and $S$, and therefore generates entanglement between $0$ and $S$ for some product input state.
\end{proof}

The Abelian case thus admits a simple qualitative classification into (i) symmetry-commuting gates that remain local, (ii) covariant gates that acquire only frame-dependent phases on the control, and (iii) generic gates whose orbit under the group action makes the transformed operation genuinely entangling between the old frame and the target system. This is the group-theoretic underpinning of the `remains local vs becomes controlled vs entangling' trichotomy used in the main text.

\subsection{Specialisation to $G=\mathbb{Z}_2$}
\label{app:finiteG-Z2}

For the $\mathbb{Z}_2$ model used in the main text, we have $G=\{0,1\}$ with addition modulo $2$ and
\begin{equation}
U_R(0)=\openone,\qquad U_R(1)=X,
\end{equation}
where $X$ is the Pauli-$X$ operator on the relevant qubit. The projectors on the control system are $\Pi^{(0)}_0=\ket{0}\!\bra{0}_0$ and $\Pi^{(0)}_1=\ket{1}\!\bra{1}_0$. The gate-transform theorem \eqref{eq:US-i-general} then reduces to
\begin{equation}
U_S^{(i)}
=
\ket{0}\!\bra{0}_0 \otimes U_S \otimes \openone_{\overline{\{0,S\}}}
+
\ket{1}\!\bra{1}_0 \otimes (X_S U_S X_S) \otimes \openone_{\overline{\{0,S\}}},
\label{eq:US-i-Z2}
\end{equation}
where $\openone_{\overline{\{0,S\}}}$ denotes the identity on all systems other than $0$ and $S$. This is exactly the universal single-qubit transform law used in Eq.~\eqref{eq:Z2-gate-transform} of the main text. In particular:
\begin{itemize}
\item If $U_S$ commutes with $X_S$, then $X_S U_S X_S = U_S$, and the two blocks in \eqref{eq:US-i-Z2} coincide: the gate remains local and decouples from the reference frame.
\item If $U_S$ anticommutes with $X_S$, then $X_S U_S X_S = -U_S$, and the frame change produces a controlled phase between the $\ket{0}_0$ and $\ket{1}_0$ sectors.
\item For a generic gate such as the Hadamard $H$, for which neither $[H,X]=0$ nor $\{H,X\}=0$, the two blocks $H$ and $XHX$ are not proportional, and the transformed gate is genuinely entangling between system $0$ and system $S$.
\end{itemize}

These statements underpin the concrete $\mathbb{Z}_2$ circuit examples analysed in Secs.~\ref{sec:finiteG-gate-calculus}and ~\ref{sec:circuits}, where we show how Hadamard and CNOT gates transform under changes of quantum reference frame and how this induces a redistribution of coherence and entanglement across different internal perspectives.

\section{Coherence--concurrence complementarity for pure two-qubit states}
\label{app:coh_conc}

In this appendix, we give a self-contained derivation of the complementarity relation
\begin{equation}
  C^2 + D^2 = 1
\end{equation}
for pure two-qubit states, using exactly the coherence and entanglement measures employed in the main text (following Refs.~\cite{sun_intrinsic_relations_2017, fan_universal_complementarity_2018, zhou_mutual_restriction_2020}
.

\subsection{Definitions and basic properties}

Let $\mathcal{H}_X \cong \mathbb{C}^2$ and $\mathcal{H}_Y \cong \mathbb{C}^2$ be qubit Hilbert spaces, and let
\begin{equation}
  \rho_{XY} = \ket{\Psi}\!\bra{\Psi}
\end{equation}
be a pure state on $\mathcal{H}_X \otimes \mathcal{H}_Y$. We denote the reduced states by
\begin{equation}
  \rho_X = \mathrm{Tr}_Y(\rho_{XY}), \qquad
  \rho_Y = \mathrm{Tr}_X(\rho_{XY}).
\end{equation}
For such pure bipartite states, the nonzero eigenvalues of $\rho_X$ and $\rho_Y$ coincide. In particular, there exist $p\in[0,1]$ and an orthonormal basis $\{\ket{0}_X,\ket{1}_X\}$ of $\mathcal{H}_X$ and $\{\ket{0}_Y,\ket{1}_Y\}$ of $\mathcal{H}_Y$ such that the Schmidt decomposition reads
\begin{equation}
  \ket{\Psi}
  \;=\;
  \sqrt{p}\,\ket{0}_X\ket{0}_Y + \sqrt{1-p}\,\ket{1}_X\ket{1}_Y,
  \label{eq:schmidt}
\end{equation}
and the reduced density matrices are diagonal in this basis:
\begin{equation}
  \rho_X = p\,\ket{0}\!\bra{0}_X + (1-p)\,\ket{1}\!\bra{1}_X,
  \qquad
  \rho_Y = p\,\ket{0}\!\bra{0}_Y + (1-p)\,\ket{1}\!\bra{1}_Y.
\end{equation}
In particular, the purities are given by
\begin{equation}
  \mathrm{Tr}(\rho_X^2) = \mathrm{Tr}(\rho_Y^2)
  = p^2 + (1-p)^2.
  \label{eq:purities}
\end{equation}

We now recall the resource measures used in the main text:  For a pure two-qubit state, the (Wootters) concurrence is
\begin{equation}
  C = \sqrt{2\bigl(1 - \mathrm{Tr}(\rho_X^2)\bigr)}
  = \sqrt{2\bigl(1 - \mathrm{Tr}(\rho_Y^2)\bigr)}.
  \label{eq:C-def}
\end{equation}
This is equivalent to the usual expression $C=2\sqrt{\det\rho_X}$. For each subsystem $Z\in\{X,Y\}$, we take
\begin{equation}
  D_Z = \sqrt{2\,\mathrm{Tr}(\rho_Z^2) - 1},
  \label{eq:DZ-def}
\end{equation}
and define the total first-order coherence as
\begin{equation}
  D = \sqrt{\frac{D_X^2 + D_Y^2}{2}}.
  \label{eq:D-total-def}
\end{equation}

These definitions coincide with those used in the coherence--concurrence complementarity literature for qubits~\cite{sun_intrinsic_relations_2017, fan_universal_complementarity_2018, zhou_mutual_restriction_2020}.

\subsection{Derivation in the Schmidt basis}

Starting from the Schmidt form in Eq.~\eqref{eq:schmidt}, the reduced states have eigenvalues \(p\) and \(1-p\), so
\begin{align}
  \mathrm{Tr}(\rho_X^2)
  &= \mathrm{Tr}(\rho_Y^2)
   = p^2 + (1-p)^2 \nonumber\\
  &= 1 - 2p(1-p).
  \label{eq:purity-explicit}
\end{align}
Substituting Eq.~\eqref{eq:purity-explicit} into Eq.~\eqref{eq:C-def} gives
\begin{align}
  C^2
  &= 2\bigl(1 - \mathrm{Tr}(\rho_X^2)\bigr)
   = 2\bigl(1 - (1 - 2p(1-p))\bigr) \nonumber\\
  &= 4p(1-p),
  \label{eq:C2-explicit}
\end{align}
so that \(C = 2\sqrt{p(1-p)}\).

From Eqs.~\eqref{eq:DZ-def} and \eqref{eq:purity-explicit},
\begin{align}
  D_X^2
  &= 2\,\mathrm{Tr}(\rho_X^2) - 1
   = 2\bigl(1 - 2p(1-p)\bigr) - 1 \nonumber\\
  &= 1 - 4p(1-p),
\end{align}
and by symmetry \(D_Y^2 = D_X^2\). Hence
\begin{equation}
  D_X^2 = D_Y^2 = 1 - 4p(1-p).
  \label{eq:DX2-DY2-explicit}
\end{equation}
Using Eq.~\eqref{eq:DX2-DY2-explicit} in Eq.~\eqref{eq:D-total-def} finally yields
\begin{align}
  D^2
  &= \frac{D_X^2 + D_Y^2}{2}
   = \frac{(1 - 4p(1-p)) + (1 - 4p(1-p))}{2} \nonumber\\
  &= 1 - 4p(1-p).
  \label{eq:D2-explicit}
\end{align}

\subsection{Coherence--concurrence complementarity}

Combining Eqs.~\eqref{eq:C2-explicit} and \eqref{eq:D2-explicit} immediately gives
\begin{equation}
  C^2 + D^2
  = 4p(1-p) + \bigl(1 - 4p(1-p)\bigr)
  = 1,
\end{equation}
for every $p \in [0,1]$. Since every pure two-qubit state admits a Schmidt decomposition of the form in Eq.~\eqref{eq:schmidt}, this yields:

\begin{proposition}
\label{prop:C2plusD2}
For any pure state $\rho_{XY} = \ket{\Psi}\!\bra{\Psi}$ of two qubits $X$ and $Y$, the concurrence $C$ and total first-order coherence $D$ defined by Eqs.~\eqref{eq:C-def}--\eqref{eq:D-total-def} satisfy
\begin{equation}
  C^2 + D^2 = 1.
\end{equation}
\end{proposition}

\subsection{Basis independence and local unitaries}

The derivation above was carried out in the Schmidt basis, but the result is basis independent because both $C$ and $D$ are invariant under local unitaries. The concurrence $C$ depends only on the eigenvalues of $\rho_X$ (or $\rho_Y$), which are unchanged under transformations of the form $U_X \otimes U_Y$. Likewise, the purities $\mathrm{Tr}(\rho_X^2)$ and $\mathrm{Tr}(\rho_Y^2)$ are unitarily invariant, so $D_X$, $D_Y$, and hence $D$ are also invariant under local unitaries. Therefore $C^2 + D^2$ is a scalar quantity that depends only on the spectrum of the reduced state, and the equality $C^2 + D^2 = 1$ holds in any local basis.

In particular, when we evaluate $C^2$ and $D^2$ for the bipartite subsystems appearing in different quantum reference frames in the main text, we are always in the setting of Proposition~\ref{prop:C2plusD2}. The frame change may alter which pair of subsystems we consider, and thus the reduced states $\rho_X$ and $\rho_Y$, but for each pure bipartite cut the same complementarity relation holds. This is the precise mathematical content of the statement that the scalar ``total quantumness'' $C^2 + D^2$ is invariant under the QRF transformations implemented in our constructions.

\section{Explicit matrix representations and subspace invariance}
\label{app:Z2-matrices}

For $G = \mathbb{Z}_2$, the QRF transformation $U_{C\to A}$ is implemented by a controlled-NOT from the new frame $A$ to the other systems $B$ and $C$, followed by a SWAP exchanging $A$ and $C$:
\begin{equation}
  U_{C\to A}
  =
  \mathrm{SWAP}_{C,A} \bigl( \mathrm{CNOT}_{A\to B}\, \mathrm{CNOT}_{A\to C} \bigr).
\end{equation}
On Pauli operators, the controlled-NOTs act as $Z_A \mapsto Z_A Z_B Z_C$, while $Z_B \mapsto Z_B$ and $Z_C \mapsto Z_C$ on the targets, so that the global parity operator $P = Z_A Z_B Z_C$ is mapped to $Z_A$. The subsequent SWAP$_{C,A}$ sends $Z_A \mapsto Z_C$, and hence the even-parity subspace $\mathcal{H}_{\mathrm{phys}} = \mathrm{span}\{\ket{000}, \ket{011}, \ket{101}, \ket{110}\}$ is invariant under $U_{C\to A}$.

To make this explicit, we evaluate $U_{C\to A}$ on the physical basis. Writing the action as
\begin{equation}
  U_{C\to A}
  =
  \mathrm{SWAP}_{C,A} \circ \bigl( \mathrm{CNOT}_{A\to B}\, \mathrm{CNOT}_{A\to C} \bigr),
\end{equation}
one finds
\begin{subequations}
\begin{align}
  \ket{000} &\mapsto \ket{000},\\
  \ket{011} &\mapsto \ket{110},\\
  \ket{101} &\mapsto \ket{011},\\
  \ket{110} &\mapsto \ket{101},
\end{align}
\end{subequations}
so $U_{C\to A}$ restricts to a permutation of the basis vectors in $\mathcal{H}_{\mathrm{phys}}$. In the ordered basis $\{\ket{000}, \ket{011}, \ket{101}, \ket{110}\}$, the matrix representation is
\begin{equation}
  U_{C\to A}\big|_{\mathcal{H}_{\mathrm{phys}}}
  =
  \begin{pmatrix}
    1 & 0 & 0 & 0 \\
    0 & 0 & 1 & 0 \\
    0 & 0 & 0 & 1 \\
    0 & 1 & 0 & 0
  \end{pmatrix}.
\end{equation}
This $4\times 4$ unitary is used to compute the theoretical resource values in Table~\ref{tab:C2-D2-frames}. The transformation $U_{C\to B}$ follows by symmetry (controlling on $B$ instead of $A$ before the swap) and has an analogous permutation form on $\mathcal{H}_{\mathrm{phys}}$.

\section{From the general QRF map to the $\mathbb{Z}_2$ unitaries and operator identities}
\label{app:Z2-operator-identities}

In this appendix we derive the three-qubit $\mathbb{Z}_2$ frame-change unitaries and operator identities used in Sec.~\ref{sec:circuits} and Sec.~\ref{sec:finiteG-gate-calculus}. Starting from the general finite-group QRF unitary $U_{0\to i}$ in Eq.~\eqref{eq:finite-group-QRF-unitary} and the observable map in Eq.~\eqref{eq:obs-transform}, we focus on $G=\mathbb{Z}_2$ acting on three systems $(A,B,C)$ and obtain explicit expressions for the frame-change unitaries $U_{C\to B}$ and $U_{C\to A}$, as well as for the transformed gates $H_A^{(B)}$, $(\mathrm{CNOT}_{A\to B})^{(B)}$, $H_A^{(A)}$, and $(\mathrm{CNOT}_{A\to B})^{(A)}$.

\subsection{From the general QRF unitary to the three-qubit $\mathbb{Z}_2$ maps}
\label{app:Z2-U}

The perspective-neutral change-of-frame unitary used in the main text is
\begin{equation}
  U_{0\to i}
  =
  \mathrm{SWAP}_{0,i}
  \sum_{g\in G}
  \ket{g}\!\bra{g}_{\,i}
  \otimes \openone_{0}
  \otimes
  \bigl[U_R(g)\bigr]^{\otimes (n-2)},
  \label{eq:appC-U-general}
\end{equation}
where $U_R$ denotes the right-regular representation of $G$ on each non-reference system~\cite{Vanrietvelde2020ChangePerspective,DeLaHamette2020QRFGeneralGroups}. 

In our three-qubit model we take $G=\mathbb{Z}_2=\{0,1\}$ with group operation given by addition modulo $2$ and
\begin{equation}
  U_R(0)=\openone,\qquad U_R(1)=X,
  \label{eq:appC-Z2-reps}
\end{equation}
where $X$ is the Pauli-$X$ operator on a qubit. The systems $(A,B,C)$ each carry a copy of the regular representation; we regard $C$ as the laboratory frame and $A,B$ as systems that can serve as internal frames.

For a change of frame $C\to B$ we identify $(0,i,R)=(C,B,\{A\})$ in Eq.~\eqref{eq:appC-U-general}, which gives
\begin{equation}
  U_{C\to B}
  =
  \mathrm{SWAP}_{C,B}
  \sum_{g\in\{0,1\}}
  \ket{g}\!\bra{g}_{\,B}
  \otimes \openone_{C}
  \otimes U_R(g)_A.
  \label{eq:appC-UCB-sum}
\end{equation}
Writing $\Pi_g^{(B)}:=\ket{g}\!\bra{g}_B$ and substituting $U_R(0)=\openone$, $U_R(1)=X$ yields
\begin{equation}
  U_{C\to B}
  =
  \mathrm{SWAP}_{B,C}\,\Bigl(\Pi^{(B)}_0\otimes\openone_A+\Pi^{(B)}_1\otimes X_A\Bigr),
  \label{eq:appC-UCB-final}
\end{equation}
which is the form used in the three-qubit circuit constructions.

Similarly, for the change of frame $C\to A$ we identify $(0,i,R)=(C,A,\{B\})$ and obtain
\begin{equation}
  U_{C\to A}
  =
  \mathrm{SWAP}_{A,C}
  \sum_{g\in\{0,1\}}
  \ket{g}\!\bra{g}_{\,A}
  \otimes \openone_{C}
  \otimes U_R(g)_B
  =
  \mathrm{SWAP}_{A,C}\,\Bigl(\Pi^{(A)}_0\otimes\openone_B+\Pi^{(A)}_1\otimes X_B\Bigr),
  \label{eq:appC-UCA-final}
\end{equation}
which is the analogous unitary for changing from frame $C$ to frame $A$. These two maps are the only frame-change unitaries required in the three-qubit examples discussed in the main text.

\subsection{Heisenberg-picture operator map}
\label{app:Z2-Heisenberg}

The observable map between frames is given by the Heisenberg-picture conjugation rule
\begin{equation}
  Z^{(\mathrm{new})} = U\,Z^{(\mathrm{old})}\,U^\dagger,
  \label{eq:appC-heisenberg}
\end{equation}
with \( U \) one of the frame-change unitaries introduced above. This is the specialisation of Eq.~\eqref{eq:obs-transform} to our finite-dimensional setting. In what follows we consider the descriptions in \( B \)'s frame and in \( A \)'s frame, corresponding to the choices \( U = U_{C\to B} \) and \( U = U_{C\to A} \), respectively. All operator identities derived below, Eqs.~\eqref{eq:appC-HA-B-final}–\eqref{eq:appC-CNOT-A-final}, result from applying Eq.~\eqref{eq:appC-heisenberg} to the elementary gates \( H_A \) and \( \mathrm{CNOT}_{A\to B} \).

In each case it is convenient to factor the frame-change unitaries as
\begin{equation}
  U_{C\to B} = S_{BC} W_{BA}, \qquad
  U_{C\to A} = S_{AC} W_{AB},
  \label{eq:appC-factorization}
\end{equation}
where \( S_{XY} \) is the SWAP between systems \( X \) and \( Y \), and
\begin{equation}
  W_{BA} = \Pi^{(B)}_0\otimes\openone_A + \Pi^{(B)}_1\otimes X_A, \qquad
  W_{AB} = \Pi^{(A)}_0\otimes\openone_B + \Pi^{(A)}_1\otimes X_B.
  \label{eq:appC-W-defs}
\end{equation}
Because the single-qubit gates under consideration act trivially on the control system and any spectator, they commute with the SWAP, and the nontrivial part of the calculation reduces to conjugation by \( W_{BA} \) or \( W_{AB} \).

\subsection{Derivation of \( H_A^{(B)} \) and \( (\mathrm{CNOT}_{A\to B})^{(B)} \)}
\label{app:Z2-B-frame}

\paragraph{Hadamard on \( A \) in \( B \)'s frame.}

We start from a Hadamard on \( A \) in \( C \)'s frame, \( H_A^{(C)} \equiv H_A \otimes \openone_{BC} \). The corresponding operator in \( B \)'s frame is
\begin{equation}
  H_A^{(B)} = U_{C\to B}\,H_A^{(C)}\,U_{C\to B}^\dagger
  = S_{BC}\,\bigl(W_{BA} H_A W_{BA}^\dagger\bigr)\,S_{BC}^\dagger,
  \label{eq:appC-HA-B-setup}
\end{equation}
where we used \( U_{C\to B}^\dagger = W_{BA}^\dagger S_{BC} \) and the fact that \( H_A \) commutes with \( S_{BC} \).

Using the projector decomposition of \( W_{BA} \),
\begin{equation}
  W_{BA}
  = \sum_{g\in\{0,1\}} \Pi^{(B)}_g\otimes U_R(g)_A
  = \Pi^{(B)}_0\otimes\openone_A + \Pi^{(B)}_1\otimes X_A,
  \label{eq:appC-WBA-sum}
\end{equation}
we obtain
\begin{align}
  W_{BA} H_A W_{BA}^\dagger
  &=
  \sum_{g,h\in\{0,1\}} \Pi^{(B)}_g \Pi^{(B)}_h \otimes U_R(g)_A H_A U_R(h)_A^\dagger \nonumber\\
  &=
  \sum_{g\in\{0,1\}} \Pi^{(B)}_g \otimes U_R(g)_A H_A U_R(g)_A^\dagger \nonumber\\
  &=
  \Pi^{(B)}_0\otimes H_A + \Pi^{(B)}_1\otimes X_A H_A X_A.
  \label{eq:appC-WBA-conjugation}
\end{align}
Conjugation by \( S_{BC} \) simply exchanges the labels \( B \leftrightarrow C \) in the projectors, so that
\begin{equation}
  H_A^{(B)} = \ket{0}\!\bra{0}_C\otimes H_A + \ket{1}\!\bra{1}_C\otimes (X_A H_A X_A),
  \label{eq:appC-HA-B-final}
\end{equation}
which coincides with the expression used in the main text. In words, in \( B \)'s frame the Hadamard on \( A \) becomes a \( C \)-controlled single-qubit gate whose branches are related by conjugation with \( X_A \).

\paragraph{CNOT from \( A \) to \( B \) in \( B \)'s frame.}

For the two-qubit gate
\begin{equation}
  \mathrm{CNOT}_{A\to B} = \Pi^{(A)}_0\otimes\openone_B + \Pi^{(A)}_1\otimes X_B,
  \label{eq:appC-CNOT-def}
\end{equation}
the image in \( B \)'s frame is
\begin{equation}
  \bigl(\mathrm{CNOT}_{A\to B}\bigr)^{(B)}
  = U_{C\to B}\,\mathrm{CNOT}_{A\to B}\,U_{C\to B}^\dagger
  = S_{BC}\,\bigl(W_{BA}\,\mathrm{CNOT}_{A\to B}\,W_{BA}^\dagger\bigr)\,S_{BC}^\dagger.
  \label{eq:appC-CNOT-B-setup}
\end{equation}
To determine \( W_{BA}\,\mathrm{CNOT}_{A\to B}\,W_{BA}^\dagger \), it is convenient to act on a basis state \( \ket{a}_A\ket{b}_B \) with \( a,b\in\{0,1\} \), suppressing the spectator \( C \). First \( W_{BA} \) maps
\[
  \ket{a}_A\ket{b}_B \longmapsto \ket{a\oplus b}_A\ket{b}_B,
\]
because \( X_A \) is applied iff \( b=1 \). Then \( \mathrm{CNOT}_{A\to B} \) flips \( b \) iff \( a\oplus b=1 \), giving
\[
  \ket{a\oplus b}_A\ket{b}_B \longmapsto \ket{a\oplus b}_A\ket{a}_B.
\]
Finally, \( W_{BA}^\dagger = W_{BA} \) is applied once more, flipping \( A \) iff the (new) value of \( B \) is \( a \):
\[
  \ket{a\oplus b}_A\ket{a}_B \longmapsto \ket{b}_A\ket{a}_B.
\]
Thus the net effect of \( W_{BA}\,\mathrm{CNOT}_{A\to B}\,W_{BA}^\dagger \) is to swap \( A \) and \( B \),
\begin{equation}
  W_{BA}\,\mathrm{CNOT}_{A\to B}\,W_{BA}^\dagger = \mathrm{SWAP}_{A,B}.
  \label{eq:appC-CNOT-is-SWAP}
\end{equation}
Conjugation by \( S_{BC} \) then relabels \( B \leftrightarrow C \), so that
\begin{equation}
  \bigl(\mathrm{CNOT}_{A\to B}\bigr)^{(B)} = \mathrm{SWAP}_{A,C},
  \label{eq:appC-CNOT-B-final}
\end{equation}
as used in the main text. In \( B \)'s frame, the interaction between \( A \) and \( B \) is therefore seen as an interaction between \( A \) and the original laboratory frame \( C \).

\subsection{Derivation of $H_A^{(A)}$ and $(\mathrm{CNOT}_{A\to B})^{(A)}$}
\label{app:Z2-A-frame}

We take $U_{C\to A} = S_{AC} W_{AB}$ and write
\begin{equation}
  H_A^{(A)} = U_{C\to A}\,H_A^{(C)}\,U_{C\to A}^\dagger
  = S_{AC}\,\bigl(W_{AB}(H_A\otimes\openone_B)W_{AB}^\dagger\bigr)\,S_{AC}^\dagger,
  \label{eq:appC-HA-A-setup}
\end{equation}
where $H_A^{(C)} = H_A\otimes\openone_B\otimes\openone_C$. Expanding in the $A$ basis gives
\begin{equation}
  H_A
  = \tfrac{1}{\sqrt{2}}\Bigl(
    \ket{0}\!\bra{0}_A
    + \ket{0}\!\bra{1}_A
    + \ket{1}\!\bra{0}_A
    - \ket{1}\!\bra{1}_A
  \Bigr).
  \label{eq:appC-HA-expansion}
\end{equation}
Using
\begin{equation}
  W_{AB}
  = \Pi^{(A)}_0\otimes\openone_B + \Pi^{(A)}_1\otimes X_B,
  \label{eq:appC-WAB-def}
\end{equation}
one finds
\begin{align}
  W_{AB}(H_A\otimes\openone_B)W_{AB}^\dagger
  &=
  \tfrac{1}{\sqrt{2}}\Bigl[
    (\ket{0}\!\bra{0}_A - \ket{1}\!\bra{1}_A)\otimes \openone_B
    + (\ket{0}\!\bra{1}_A + \ket{1}\!\bra{0}_A)\otimes X_B
  \Bigr].
  \label{eq:appC-WAB-conjugation}
\end{align}
Conjugation by $S_{AC}$ swaps $A$ and $C$ in the projectors, so that
\begin{equation}
  H_A^{(A)}
  =
  \tfrac{1}{\sqrt{2}}\Bigl[
    (\ket{0}\!\bra{0}_C - \ket{1}\!\bra{1}_C)\otimes \openone_B
    + (\ket{0}\!\bra{1}_C + \ket{1}\!\bra{0}_C)\otimes X_B
  \Bigr],
  \label{eq:appC-HA-A-final}
\end{equation}
which is Eq.~\eqref{eq:appC-HA-A-final} . In $A$’s frame, the Hadamard on $A$ is thus seen as a two-qubit gate $V_{BC}$ that entangles $B$ and $C$.

For the CNOT we start from
\begin{equation}
  \mathrm{CNOT}_{A\to B}
  = \Pi^{(A)}_0\otimes\openone_B + \Pi^{(A)}_1\otimes X_B.
  \label{eq:appC-CNOT-A-start}
\end{equation}
In this case $W_{AB}$ coincides with $\mathrm{CNOT}_{A\to B}$,
\begin{equation}
  W_{AB} = \mathrm{CNOT}_{A\to B},
  \label{eq:appC-WAB-is-CNOT}
\end{equation}
so $W_{AB}$ commutes with $\mathrm{CNOT}_{A\to B}$ and
\begin{equation}
  W_{AB}\,\mathrm{CNOT}_{A\to B}\,W_{AB}^\dagger
  = \mathrm{CNOT}_{A\to B}.
  \label{eq:appC-CNOT-commutes}
\end{equation}
The transformed operator in $A$’s frame is therefore obtained by swapping $A$ and $C$:
\begin{equation}
  \bigl(\mathrm{CNOT}_{A\to B}\bigr)^{(A)}
  = S_{AC}\,\mathrm{CNOT}_{A\to B}\,S_{AC}^\dagger
  = \mathrm{CNOT}_{C\to B},
  \label{eq:appC-CNOT-A-final}
\end{equation}
which reproduces Eq.~\eqref{eq:appC-CNOT-A-final} in the main text. Thus in $A$’s frame the two-qubit interaction is between $C$ and $B$, with $C$ as the control.

\end{document}